\documentclass[onecolumn,draftclsnofoot,12pt]{IEEEtran}
\usepackage{amsfonts}
\usepackage{amsmath}
\usepackage{amsthm}
\usepackage{verbatim}
\usepackage{setspace}
\usepackage[caption=false,font=footnotesize]{subfig}

\newtheorem{theorem}{Theorem}
\newtheorem{lemma}{Lemma}
\newtheorem{proposition}{Proposition}
\newtheorem{corollary}{Corollary}

\newtheorem{remark}{Remark}

\usepackage{cite}

\ifCLASSINFOpdf
   \usepackage[pdftex]{graphicx}
\else
\fi

\def\bb0{{\mathbb{0}}}

\def\bb{{\mathbf{b}}}

\def\b0{{\mathbf{0}}}

\def\sf0{{\mathsf{0}}}

\usepackage{algorithm}
\usepackage{algpseudocode}
\usepackage{epstopdf}
\usepackage{balance}
\usepackage{bbm}
\usepackage{todo}
\IEEEoverridecommandlockouts

\begin{document}
%
\title{Wirelessly Powered Communication Networks with Short Packets}


\author{Talha Ahmed Khan${}^{*}$, Robert W. Heath Jr.${}^{*}$ and Petar Popovski${}^{\dagger}$\\

\thanks{${}^{*}$ Talha Ahmed Khan and Robert W. Heath Jr. are with the Department of Electrical and Computer Engineering at The University of Texas at Austin, USA (Email: \{talhakhan, rheath\}@utexas.edu).}
\thanks{${}^{\dagger}$ Petar Popovski is with the Department of Electronic Systems at Aalborg University, Denmark (Email: petarp@es.aau.dk).}
\thanks{This work was supported in part by the Army Research Office under grant W911NF-14-1-0460, and a gift from Mitsubishi Electric Research Labs.}}
\maketitle
\begin{abstract}
Wirelessly powered communications will entail short packets due to naturally small payloads, low-latency requirements and/or insufficient energy resources to support longer transmissions. In this paper, a wirelessly powered communication system is investigated where an energy harvesting transmitter, charged by one or more power beacons via wireless energy transfer, attempts to communicate with a receiver over a noisy channel. 
Under a save-then-transmit protocol, the system performance is characterized using metrics such as the energy supply probability at the transmitter, and the achievable rate at the receiver for the case of short packets. 
Leveraging the framework of finite-length information theory, tractable analytical expressions are derived for the considered metrics in terms of system parameters such as the harvest blocklength, the transmit blocklength, the harvested power and the transmit power.  
The analysis provides several useful design guidelines.
Though using a small transmit power or a small transmit blocklength helps avoid energy outages, the consequently smaller signal-to-noise ratio or the fewer coding opportunities may cause an information outage.  
Scaling laws are derived to capture this inherent trade-off between the harvest and transmit blocklengths.
Moreover, the asymptotically optimal transmit power is derived in closed-form. Numerical results reveal that power control is essential for improving the achievable rate of the system in the finite blocklength regime. The asymptotically optimal transmit power yields nearly optimal performance in the finite blocklength regime. 
\end{abstract}

%
\IEEEpeerreviewmaketitle

\section{Introduction}
With wireless devices getting smaller and more energy-efficient, energy harvesting is emerging as a potential technology for powering such miniature devices \cite{EnergyHarvestWirelessCommSurvey2015,RFsurveyLondon,talla2015powering,GollakotaRF,Vullers2009684}. This is attractive for future paradigms such as the Internet of Things (IoT), where powering a massive number of devices will be a major challenge\cite{IoT2014}. 
Many IoT applications will entail sensors with sporadic sensing and communication activity, resulting in an {average} power requirement on the order of microwatts to milliwatts. 
Depending on the application, the sensor may harvest energy from ambient sources such as solar, thermal, kinetic, or RF (radio frequency) waves 
\cite{RFsurveyLondon,talla2015powering,GollakotaRF,Vullers2009684,EnergyHarvestWirelessCommSurvey2015}.
Of interest to this work is RF or wireless energy harvesting, where a harvesting node extracts energy from the incident RF signals. 
This is a suitable option for ultra low-power applications because i) wireless signals are available anywhere and anytime, ii) the harvesting operation relies on a simple circuit consisting of a rectifying antenna which can be integrated with the communication circuitry in small form factors\cite{Kaibin2015Cut}, and iii) the energy delivered to the harvester can be controlled by leveraging the wireless infrastructure\cite{talla2015powering,Kaibin2015Cut}. 
In contrast to most wireless systems designed for Internet access, the energy harvesting communication systems used in IoT applications will likely feature short packets. This is due to intrinsically small data payloads, low-latency requirements, and/or lack of energy resources to support longer transmissions\cite{EnergyHarvestWirelessCommSurvey2015,durisi2015towards,yang2014finite,fong2015non}.  

For an energy harvesting system with short packets, the capacity analysis conducted in the asymptotic blocklength regime could be misleading. This has spurred research characterizing the performance of an energy harvesting communication system in the non-asymptotic or finite blocklength regime \cite{polyanskiy2010finite,yang2014finite,fong2015non,shenoy2016finite,guo2016finite,ebrahim2015lowlatency}. This line of research leverages the finite-blocklength information theoretic framework proposed in \cite{polyanskiy2010finite} (see \cite{Tan2014} for an overview). 
The work in \cite{yang2014finite} was first to investigate energy harvesting channels in the finite blocklength regime. In \cite{yang2014finite}, the non-asymptotic achievable rate was characterized for a noiseless binary communications channel with an energy harvesting transmitter. This work was extended to the case of an additive white Gaussian noise (AWGN) channel and to more general discrete memoryless channels in \cite{fong2015non}. For an energy harvesting transmitter operating under a save-then-transmit protocol (first proposed in \cite{ozel2012awgn}), a lower bound on the achievable rate at the receiver was derived in the finite blocklength regime\cite{fong2015non}. For the setup considered in \cite{fong2015non}, the work in \cite{shenoy2016finite} provided tighter bounds on the non-asymptotic achievable rate for an AWGN energy harvesting channel.
The authors in \cite{guo2016finite} investigated the mean delay of an energy harvesting channel in the finite blocklength regime.
Unlike the work in \cite{yang2014finite,fong2015non,shenoy2016finite,guo2016finite} which assume an infinite battery at the energy harvester, \cite{ebrahim2015lowlatency} conducted a finite-blocklength analysis for a battery-less energy harvesting channel.  

The capacity analysis of energy harvesting channels in the asymptotic blocklength regime has received considerable attention\cite{ozel2012awgn,ozel11zerobattery,mao13eh,jog14,ozgur15,fund15eh}. The capacity of an energy harvesting AWGN channel under stochastic energy arrivals was derived in \cite{ozel2012awgn} assuming an infinite battery at the energy harvester. For a similar setup, the capacity analysis for a battery-less energy harvester was conducted in \cite{ozel11zerobattery}.
An energy harvesting transmitter with a finite battery was considered in \cite{mao13eh}, and the capacity was analyzed using Shannon strategies for discrete memoryless channels.
The capacity of an energy harvesting AWGN channel with a finite battery was considered in \cite{jog14} for the case of deterministic energy arrivals. Also assuming a finite battery, the approximate capacity of an energy harvesting AWGN channel with Bernoulli energy arrivals was derived in \cite{ozgur15}. A comprehensive review of the capacity of energy harvesting channels is provided in \cite{fund15eh}. 

In this paper, we investigate the performance of a wireless-powered communication system where an RF energy harvesting node, charged by wireless power beacons via wireless energy transfer, attempts to communicate with a receiver over an AWGN channel.
We conduct the analysis for two cases. We first provide an analytical treatment for the case of a single power beacon. We then extend the analysis to a large-scale Poisson network with multiple power beacons. 
Using the framework of finite-length information theory\cite{polyanskiy2010finite}, we characterize the energy supply probability and the achievable rate of the considered system with short packets, i.e., in the non-asymptotic or finite blocklength regime. Leveraging the analytical results, we expose the interplay between key system parameters such as the harvest and transmit blocklengths, the average harvested power, and the transmit power. 
We analytically characterize the scaling laws for the harvest and transmit blocklengths in terms of the transmit-to-harvest power ratio and the target error probability.
We also provide closed-form analytical expressions for the asymptotically optimal transmit power. Numerical results reveal that the asymptotically optimal transmit power yields nearly optimal performance in the finite blocklength regime. We also examine how the power beacon transmit power and density impacts the overall performance.

Our work differs from the existing literature on several accounts.
First, the prior work \cite{fong2015non,yang2014finite,ebrahim2015lowlatency,guo2016finite,shenoy2016finite} on energy harvesting systems in the finite blocklength regime falls short of characterizing the performance for the case of wireless energy harvesting. Second, most prior work \cite{fong2015non,yang2014finite,ebrahim2015lowlatency,guo2016finite,ozel2012awgn,shenoy2016finite} implicitly assumes concurrent harvest and transmit operation, which may be infeasible in practice. For example, a power beacon may remain silent during the communication phase to avoid interfering with the communication link\cite{Kaibin2015Cut}. 
Third, none of these finite-blocklength analyses treats the case of multiple power beacons. 
This paper is an extension of our conference paper \cite{khan2016wirelessly}, where limited analytical results were provided for the case of a single power beacon. 

The paper is organized as follows. The system model is described in Section \ref{secSys}. The analytical characterization of the energy supply probability and the achievable rate for the case of single power beacon is presented in Section \ref{secAnl}. Section \ref{secMP} extends the analysis to include multiple power beacons. Simulation results are provided in Section \ref{secSim}. The paper is concluded in Section \ref{secConc}.
\section{System Model}\label{secSys}
We consider a wireless-powered communication system where one or more wireless power beacons (PBs) use wireless energy transfer to charge an energy harvesting (EH) node, which then attempts to communicate with another receiver (RX) using the harvested energy (see Fig. \ref{fig:sysmod}).  
The nodes are assumed to be equipped with a single antenna each. 
We present an analytical treatment for two cases: i) the energy harvesting node is powered by a single power beacon, and ii) the energy harvesting node is powered by a large-scale network consisting of multiple power beacons. 
We now describe the system model for the case of a single power beacon. Any additional description for the case of multiple power beacons will be provided in Section \ref{secMP}. 
We assume that the energy harvester uses a \emph{save-then-transmit} protocol \cite{ozel2012awgn} to enable wireless-powered communications. 
The considered protocol divides the communication frame consisting of $S$ channel uses (or slots) into an energy harvesting phase having $m$ channel uses, and an information transmission phase having $n$ channel uses. 
The first $m$ channel uses are used for harvesting energy from the RF signals transmitted by the power beacon, which is then saved in a (sufficiently large) energy buffer.
This is followed by an information transmission phase consisting of $n$ channel uses, where the transmitter uses the harvested energy to transmit information to the receiver. We call $m$ the \textit{harvest} blocklength, $n$ the \textit{transmit} blocklength, and $S=m+n$ the \textit{total} blocklength or frame size. 
We will conduct the subsequent analysis for the non-asymptotic blocklength regime, i.e., for the practical case of \emph{short packets} where the total blocklength is finite.
\begin{figure} [t]
	\centerline{	\includegraphics[width=.7\columnwidth]{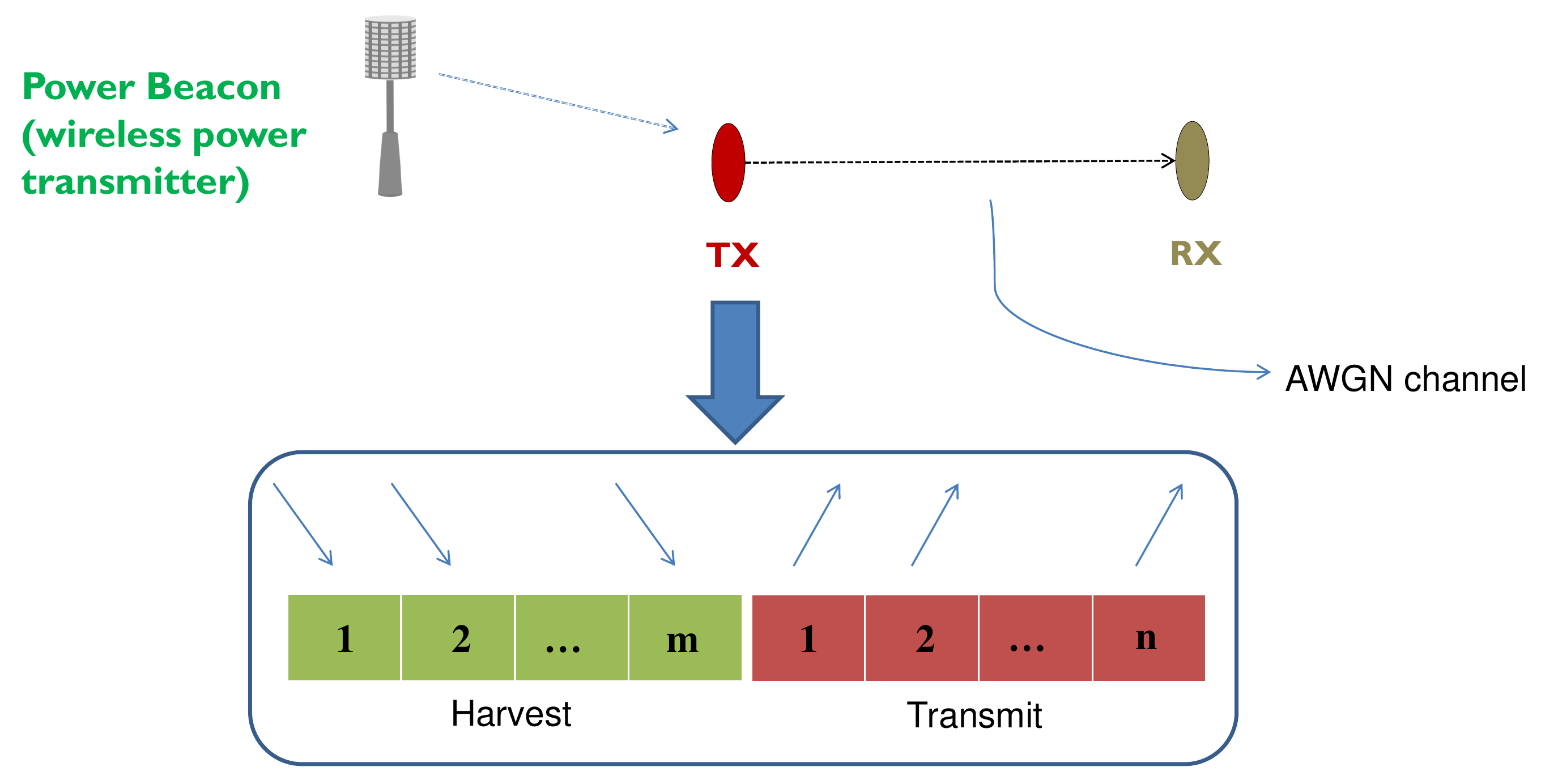}}
	\caption{A single power beacon charges an energy harvesting node, which operates under a save-then-transmit protocol to communicate with its desired receiver.}
	\label{fig:sysmod}
\end{figure}

\subsection{Energy Harvesting Phase} 
The signal transmitted by a power beacon experiences distance-dependent path loss and channel fading before reaching the energy harvesting node. The harvested energy is, therefore, a random quantity due to the underlying randomness of the wireless link. 
We let random variable $Z_i=\frac{\mu}{\ell(r,\eta) }P_{\rm{PB}} H_i$ model the energy (or power) harvested in slot $i$ ($i=1,\cdots,m$), where $\mu\in(0,1]$ denotes the conversion efficiency of the energy harvester, $P_{\rm{PB}}$ is the PB transmit power (i.e., energy per PB symbol), $\ell(r,\eta)$ gives the average large-scale path loss given a PB-EH link distance $r$ and path loss exponent $\eta>2$, while the random variable $H_i$ denotes the small-scale channel gain. Note that we have ignored the energy due to noise since it is negligibly small. 
We consider quasi-static block flat Rayleigh fading for the PB-EH links such that the channel remains constant over (the harvesting phase of) a frame, and randomly changes to a new value for the next frame. 
In other words, the energy arrivals within a harvesting phase are fully correlated, i.e., $Z_i=Z_1\equiv Z,\,\forall\,i=1,2,\cdots,m$, where $Z_i$ is exponentially distributed with mean $\mathbb{E}[Z_i]\triangleq P_{\rm{E}}=\frac{\mu}{\ell(r,\eta)} P_{\rm{PB}}$.
This is motivated by the observation that the harvest blocklength in a \emph{short-packet} communication system would typically be smaller than the channel coherence time. 
\subsection{Information Transmission Phase} The energy harvesting phase is followed by an information transmission phase where the EH node attempts to communicate with a destination RX node over an unreliable AWGN channel.
Contrary to the harvesting operation, here noise plays a significant role. 
We assume that the EH node uses a Gaussian codebook for signal transmission (see Section \ref{secIT}).
We let $X_\ell$ be the signal intended for transmission in slot $\ell$ with average power $P_{\rm{t}}$, where $\ell=1,\cdots,n$, and $n$ is fixed. 
The resulting (intended) sequence $X^n=\left(X_1,\cdots,X_n\right)$ consists of independent and identically distributed (IID) Gaussian random variables such that $X_\ell$ {\raise.17ex\hbox{$\scriptstyle{\sim}$}}~$\mathcal{N}(0,P_{\rm{t}})$. 
To transmit the intended sequence $X^n$ over the transmit blocklength, the EH node needs to satisfy the following energy constraints. 
\begin{align}\label{eq:constraint}
	\sum_{\ell=1}^{k}X_\ell^2&\leq\sum_{i=1}^{m}Z_i\qquad k=1,2,\cdots,n.
\end{align}
The following lemma simplifies the multiple energy constraints into a single constraint.
\begin{lemma}\normalfont\label{lemma:energy} For a random sequence $\{X_\ell\}_{\ell=1}^{n}$ for the transmit phase, and a random energy sequence $\{Z_i\}_{i=1}^{m}$ for the harvest phase, the probability of violating the energy constraints in (\ref{eq:constraint}) is given by
	\begin{align}
	\Pr\left[{\bigcup_{k=1}^{n}}\left\{\sum_{\ell=1}^{k}X_\ell^2\leq\sum_{i=1}^{m}Z_i\right\} \right]&=1-
	\Pr\left[\sum_{\ell=1}^{n}X_\ell^2\leq\sum_{i=1}^{m}Z_i \right].
	\end{align}

\begin{proof}
	The result follows by noting that 
	\begin{align}
	\Pr\left[{\bigcap_{k=1}^{n}}\left\{\sum_{\ell=1}^{k}X_\ell^2\leq\sum_{i=1}^{m}Z_i\right\} \right]&=
	\Pr\left[\left\{\sum_{\ell=1}^{n}X_\ell^2\leq\sum_{i=1}^{m}Z_i\right\} \right]\nonumber\\
	&\quad\underbrace{\times\Pr\left[{\bigcap_{k=1}^{n-1}}\left\{\sum_{\ell=1}^{k}X_\ell^2\leq\sum_{i=1}^{m}Z_i\right\}\bigg| {\sum_{\ell=1}^{n}X_\ell^2\leq\sum_{i=1}^{m}Z_i}\right]}_{1}\nonumber
	\end{align}
\end{proof}
\end{lemma}
Using Lemma \ref{lemma:energy}, the constraints in (\ref{eq:constraint}) further simplify to $\sum_{\ell=1}^{n}X_\ell^2\leq mZ$ for the case of correlated energy arrivals. We let $\tilde{X}^n=\left(\tilde{X}_1,\cdots,\tilde{X}_n\right)$ be the transmitted sequence. Note that $\tilde{X}^n\neq{X}^n$ when the energy constraints are violated as the EH node lacks sufficient energy to put the intended symbols on the channel.
The signal received at the destination node in slot $\ell$ is given by
$Y_\ell=\tilde{X_\ell} + V_\ell$, where $V^n=\left(V_1,\cdots,V_n\right)$ is an IID sequence modeling the receiver noise such that $V_\ell~{\raise.17ex\hbox{$\scriptstyle{\sim}$}}~\mathcal{N}(0,\sigma^2)$ is a zero-mean Gaussian random variable with variance $\sigma^2$. Note that any deterministic channel attenuation for the EH-RX link can be equivalently tackled by scaling the noise variance. Similarly, we define $Y^n=\left(Y_1,\cdots,Y_n\right)$ as the received sequence.

\subsection{Information Theoretic Preliminaries}\label{secIT}
We now describe the information theoretic preliminaries for the EH-RX link. Let us assume that the EH node transmits a message $W\in\mathcal{W}$ 
over $n$ channel uses. Assuming $W$ is drawn uniformly from $\mathcal{W}\triangleq\{1,2,\cdots,M\}$, we define an $(n,M)$-code having the following features:
It uses a set of encoding functions $\{\mathcal{F}_\ell\}_{\ell=1}^{n}$ for encoding the source message $W\in\mathcal{W}$ given the energy harvesting constraints, i.e., the source node uses $\mathcal{F}_\ell : \mathcal{W}\times\mathbb{R}^\ell_+ \rightarrow \mathbb{R}$ for transmission slot $\ell$, where
$\mathcal{F}_\ell(W,Z^\ell)=\tilde{X}_\ell$ given $Z^\ell=(Z_1,\cdots,Z_\ell)$ such that the energy harvesting constraint in (\ref{eq:constraint}) is satisfied. Specifically, $\tilde{X}_\ell={X}_\ell$ where ${X}_\ell~${\raise.17ex\hbox{$\scriptstyle{\sim}$}}~$\mathcal{N}(0,P_{\rm{t}})$ is drawn IID from a Gaussian codebook when (\ref{eq:constraint}) is satisfied, and $\tilde{X}_\ell=0$ otherwise.
It uses a decoding function $\mathcal{G}:\mathbb{R}^n\rightarrow\mathcal{W}$ that produces the output $\mathcal{G}(Y^n)=\hat{{W}}$, where $Y^n=\left(Y_1,\cdots,Y_n\right)$ is the sequence received at the destination node.

We let $\epsilon\in [0,1)$ denote the target error probability for the noisy communication link.
For $\epsilon\in[0,1)$, an $(n,M,\epsilon)$-code for an AWGN EH channel is defined as the $(n,M)$-code for an AWGN channel such that the average probability of decoding error $\Pr\{\hat{W}\neq W\}$ does not exceed $\epsilon$.
A rate $R$ is \emph{$\epsilon$-achievable} for an AWGN EH channel if there exists a sequence of {$(n,M_n,\epsilon_n)$-codes} such that $\liminf\limits_{n\rightarrow\infty}\frac{1}{n}\log(M_n) \geq R$ and $\limsup\limits_{n\rightarrow\infty} \epsilon_n \leq \epsilon$.
The \emph{$\epsilon$-capacity} $C_\epsilon$ for an AWGN EH channel is defined as $C_\epsilon=\sup\{R:R\,\,\text{is}\, \epsilon\text{-\textit{achievable}}\}$. 	
\subsection{Performance Metrics} 
We now introduce the metrics used for characterizing the performance of the considered \emph{short-packet} wireless-powered communications system. 
Note that the overall performance is marred by two key events. First, due to lack of sufficient energy, the EH node may not be able to transmit the intended codewords during the information transmission phase, possibly causing a decoding error at the receiver. Second, due to  a noisy EH-RX channel, the received signal may not be correctly decoded. For the former, we define a metric called the \emph{energy supply probability}, namely, the probability $\Pr\left[\sum_{i=1}^{n}X_i^2\leq mZ\right]$ that an EH node can support the intended transmission. 
For the latter, we define and characterize the \emph{$\epsilon$-achievable rate} in the finite blocklength regime.

\section{Single Power Beacon}\label{secAnl}
In this section, we characterize the energy supply probability and the achievable rate in the finite blocklength regime for an energy harvester powered by a single power beacon. We also provide closed-form analytical expressions for the optimal transmit power.
\subsection{Energy Supply Probability}
We define the \textit{energy supply probability} $P_\textrm{es}(m,n,a)$ as the probability that an EH node has sufficient energy to transmit the intended codeword, namely, 
\begin{align}\label{eq:def}
P_\textrm{es}(m,n,a)=\Pr\left[\sum_{i=1}^{n}X_i^2\leq mZ\right]
\end{align}
for a harvest blocklength $m$, a transmit blocklength $n$, and a power ratio $a=\frac{P_{\rm{t}}}{P_{\rm{E}}}$. Similarly, we define $P_\textrm{eo}(m,n,a)=1-P_\textrm{es}(m,n,a)$  as the \textit{energy outage probability} at the energy harvesting node. The following proposition characterizes the energy supply probability for the considered system.

\begin{proposition}\normalfont
Assuming the intended transmit symbols $\{X_i\}_{i=1}^{n}$ are drawn IID from $\mathcal{N}(0,P_t)$, the energy sequence $\{Z_i\}_{i=1}^{m}=Z$ is fully correlated, and $Z$ follows an exponential law with mean $P_{\rm{E}}$, the energy supply probability is given by
\begin{align}\label{eq:p_ec single PB}
	P_\textrm{es}(m,n,a)=\frac{1}{\left(1+\frac{2a}{m}\right)^{\frac{n}{2}}}
\end{align}
for $m>2a$ where $a=\frac{P_\textrm{t}}{P_{\rm{E}}}$, while $m$ and $n$ denote the blocklengths for the harvest and the transmit phase.
\end{proposition}
\begin{proof}
	The proof follows by leveraging the statistical properties of the random variables. Consider
	\begin{align}\label{eq:prop 1}
	P_\textrm{es}\left(m,n,a\right)&=\Pr\left[\sum_{i=1}^{n}X_i^2\leq mZ\right]
	\overset{(a)}{=}\Pr\left[W \leq \frac{mZ}{P_\textrm{t}}\right]\nonumber\\
	&\overset{(b)}{=}\mathbb{E}_{W}\left[e^{-\frac{P_\textrm{t}}{P_{\rm{E}}m}W}\right]
	\overset{}{=}\frac{1}{\left(1+\frac{2a}{m}\right)^\frac{n}{2}}
	\end{align}
	where (${a}$) follows from the substitution $W=\frac{1}{P_\textrm{t}}\sum_{i=1}^{n}X_i^2$ where $W$ is a Chi-squared random variable with $n$ degrees of freedom, and (${b}$) is obtained by conditioning on the random variable $W$, and by further noting that $Z$ is exponentially distributed with mean $P_{\rm{E}}$. Assuming $m>2a$, the last equation follows from the definition of the moment generating function of a Chi-squared random variable.
\end{proof}
While Proposition 1 is valid for $m>2a$, we note that this is the case of practical interest since it is desirable to operate at $a<1$, as evident from Section \ref{secSim}. Further, the expression in (\ref{eq:p_ec single PB}) makes intuitive sense as the energy outages would increase with the transmit blocklength $n$ for a given $m$, and decrease with the harvest blocklength $m$ for a given $n$. 
Let us fix $P_{\rm{t}}$ and $P_{\rm{E}}$. For a given $m$, we may improve the reliability of the EH-RX communication link by increasing the blocklength $n$, albeit at the expense of the energy supply probability. 
With a smaller transmit power $P_{\rm{t}}$, the energy harvester is less likely to run out of energy during an ongoing transmission. Therefore, when $m+n$ is fixed, we may reduce $P_{\rm{t}}$ to meet the energy supply constraint, but this would reduce the channel signal-to-noise ratio (SNR). 
This underlying tension between the energy availability and the communication reliability will be highlighted throughout the rest of this paper.
The following discussion relates the transmit power to the harvest and transmit blocklengths, illustrating some of the key tradeoffs.  
\begin{remark}\normalfont
The energy supply probability is more sensitive to the length of the transmit phase compared to that of the harvest phase. 
This observation also manifests itself in terms of the energy requirements at the transmitter. For instance, to maintain an energy supply probability $\rho$, it follows from (\ref{eq:p_ec single PB}) that the power ratio satisfies $a\geq\frac{m}{2}\left(\rho^{-\frac{2}{n}}-1\right)$. Note that the power ratio varies only linearly with the harvest blocklength $m$, but superlinearly with the transmit blocklength $n$. This further implies that for a fixed $n$, doubling the harvest blocklength relaxes the transmit power budget by the same amount. That is, the energy harvester can double its transmit power $P_{\rm{t}}$ (and therefore the channel SNR) without violating the required energy constraints. In contrast, reducing the transmit blocklength for a given $m$ brings about an exponential increase in the transmit power budget at the energy harvester. 
\end{remark}
The following corollary treats the scaling behavior of the energy supply probability as the  blocklength becomes large.
\begin{corollary}\normalfont
When the harvest blocklength $m$ scales in proportion to the transmit blocklength $n$ such that $m=cn$ for some constant $c>0$, the energy supply probability $P_\textrm{es}(m,n,a)$ converges to a limit as $n$ becomes asymptotically large. In other words, $\lim\limits_{n\rightarrow\infty}P_\textrm{es}(m,n,a)=e^{-\frac{a}{c}}<1$ such that the limit only depends on the power ratio $a>0$ and the proportionality constant $c>0$. Further, under proportional blocklength scaling, this limit also serves as an upper bound on the energy supply probability for finite blocklengths, i.e., $P_\textrm{es}(m,n,a)\leq e^{-\frac{a}{c}}<1$.
\end{corollary}
The previous corollary also shows that energy outage is a fundamental bottleneck regardless of the blocklength, assuming at best linear scaling.

\subsection{Achievable Rate}
The following result characterizes the $\epsilon$-achievable rate of the considered wireless-powered communication system in the finite blocklength regime.

\begin{theorem}\normalfont
Given a target error probability $\epsilon\in[0,1)$ for the noisy channel, the $\epsilon$-achievable rate $R_{\rm{EH}}^{}\left(\epsilon,m,n,a,\gamma\right)$ of the considered system with harvest blocklength $m$, transmit blocklength $n$, power ratio $a$ (where $2a<m$), and the SNR $\gamma=\frac{P_{\rm{t}}}{\sigma^2}$ is given by 
\begin{align}\label{eq:main}
	R_{\rm{EH}}^{}\left(\epsilon,m,n,a,\gamma\right)=\frac{\frac{n\log(1+\gamma)}{2}
		-\sqrt{\frac{2+\epsilon}{\epsilon}\frac{\gamma}{\gamma+1}n}
			-{(n)}^{\frac{1}{4}}-1}
			{n+m}
\end{align}
for all tuples $(m,n)$ satisfying
\begin{align}\label{eq:mainc1}
m\geq\frac{2a}{\exp\left(\frac{2\ln(1+0.5\epsilon)}{\left(\ln\left[\frac{2+\epsilon}
		{\epsilon^2}\right]\right)^4}\right)-1}
\end{align}
and
\begin{align}\label{eq:mainc2}
n\leq 2 \frac{\ln(1+0.5\epsilon)}{\ln\left(1+\frac{2a}{m}\right)}.
\end{align}	
\end{theorem}
\begin{proof}
	See Appendix A.
\end{proof}
For a given target error probability $\epsilon$, 
a harvest blocklength $m$ can support a transmit blocklength only as large as in (\ref{eq:mainc2}). Moreover, a sufficiently large $m$, as given in (\ref{eq:mainc1}), is required for a sufficiently large $n$ to meet the target error probability $\epsilon$. The constraints in $(\ref{eq:mainc1})$ and $(\ref{eq:mainc2})$ can be equivalently written as 
\begin{align}\label{eq:mainc11}
n\geq\left[\log\left(\frac{2+\epsilon}{\epsilon^2}\right)\right]^4
\end{align}
and
\begin{align}\label{eq:mainc22}
m\geq\frac{2a}{(1+0.5\epsilon)^{\frac{2}{n}}-1}
\end{align}
A sufficiently long transmit codeword is required to meet the reliability requirements of the communication link. Similarly, a sufficiently long harvest blocklength is required to replenish the energy supply. In latency-constrained systems where the total blocklength is fixed, this interplay between the transmit and harvest blocklength results in a trade-off between the energy supply probability and the communication reliability. 
For the rest of the analysis, we assume that minimum possible blocklengths are selected to satisfy the constraints in (\ref{eq:mainc11}) and (\ref{eq:mainc22}), i.e., we set $n=\Big\lceil\left(\log\left(\frac{2+\epsilon}{\epsilon^2}\right)\right)^4\Big\rceil_{\rm{ev}}$ 
and $m=\Big\lceil\frac{2a}{(1+0.5\epsilon)^{\frac{2}{n}}-1}\Big\rceil$, where $\lceil x\rceil_{\rm{}}$ (or $\lceil x\rceil_{\rm{ev}}$) returns the smallest integer (or even integer) not smaller than $x$. We call it the \textit{minimum latency approach}. 
The following remark illustrates the scaling behavior of the harvest and transmit blocklengths.

\begin{remark}\normalfont\label{rem:scaling}
Under the minimum latency approach, the harvest blocklength scales almost linearly with the transmit blocklength according to the law $m\approx\frac{2a}{\epsilon}n$. 
This follows from the constraint in (\ref{eq:mainc2}) where $m=\frac{2a}{\left[1+0.5\epsilon\right]^{\frac{2}{n}-1}}\approx\frac{2a}{\epsilon}n$ when $\epsilon$ is small. 
Further, the scaling rate $\frac{m}{n}$ is directly proportional to the power ratio $a$ and inversely proportional to the error $\epsilon$. For example, fix $n$ and $a$. A $k$-fold reduction in $\epsilon$ requires a $k$-fold increase in the harvest blocklength to attain the corresponding $\epsilon$-achievable rate.   
This increase in reliability, however, comes at the expense of a reduced rate and an increased latency since the harvesting overhead is $1+\frac{2a}{\epsilon}$ and the total blocklength grows as $(1+\frac{2a}{\epsilon})n$. This further suggests that we may overcome the rate (and latency) loss by a $k$-fold increase in $a$, i.e., by increasing $P_{\rm{E}}$ for a fixed $P_{\rm{t}}$. This could be achieved by increasing the PB transmit power and/or improving the rectifier efficiency. 
\end{remark}

The following proposition provides an analytical expression for the achievable rate in the asymptotic blocklength regime. We note that the asymptotic results provide a useful analytical handle for the non-asymptotic case as well. 
\begin{proposition}\normalfont\label{proposition:asym rate}
Let $R_{\textrm{EH}}^{\infty}(\epsilon,a,\gamma)$ denote the asymptotic achievable rate as the transmit blocklength $n\rightarrow\infty$ (and consequently the harvest blocklength $m\rightarrow\infty$ under the minimum latency approach), i.e., $R^{\infty}_{\textrm{EH}}(\epsilon,a,\gamma)=\lim\limits_{n\rightarrow\infty}R_\textrm{EH}(\epsilon,m,n,a,\gamma)$. It is given by
\begin{align} \label{eq:prop2}
R^\infty_{\textrm{EH}}(\epsilon,a,\gamma)&= L(a,\epsilon) C^\infty_{\textrm{AWGN}}(\gamma)
\end{align}
where \begin{equation}C^\infty_{\textrm{AWGN}}(\gamma)=\frac{1}{2}\log(1+\gamma),\quad\gamma\geq 0 
\end{equation}
denotes the capacity of an AWGN channel without the energy harvesting constraints, whereas
\begin{align}\label{eq:prpo2a}
L(a,\epsilon)=\frac{1}{1+\frac{a}{\log\left(1+0.5\epsilon\right)}},\quad a\geq 0,\,\epsilon\in[0,1)
\end{align}
where $L(a,\epsilon)\in[0,1]$ such that $1-L(a,\epsilon)$ gives the (fractional) loss in capacity due to energy harvesting constraints.
\end{proposition}
\begin{proof}
Using (\ref{eq:main}), $R_{\textrm{EH}}^{\infty}\left(\epsilon,a,\gamma\right)$ can be expressed as  
\allowdisplaybreaks
\begin{align}
	R_{\textrm{EH}}^{\infty}\left(\epsilon,a,\gamma\right)&=\lim\limits_{n\rightarrow \infty}\frac{\frac{n\log(1+\gamma)}{2}
		-\sqrt{\frac{2+\epsilon}{\epsilon}\frac{\gamma}{\gamma+1}n}
			-{(n)}^{\frac{1}{4}}-1}
			{n+m}\\
	&\overset{(a)}{=}\lim\limits_{n\rightarrow \infty}\frac{1}{1+\frac{m}{n}}\frac{\log(1+\gamma)}{2}\\
	&\overset{(b)}{=}\lim\limits_{n\rightarrow \infty}\frac{1}{1+\frac{2a}{n[1+0.5\epsilon]^{\frac{2}{n}}-1}}\frac{\log(1+\gamma)}{2}\\
	&\overset{(c)}{=}\underbrace{\frac{1}{1+\frac{a}{\log\left(1+0.5\epsilon\right)}}}_{L(a,\epsilon)}
	\underbrace{\frac{\log(1+\gamma)}{2}}_{C_{\textrm{AWGN}}^\infty(\gamma)}
\end{align}
where (${a}$) follows since the higher order terms in (\ref{eq:main}) vanish as $n\rightarrow\infty$. Note that for a given $\epsilon$ and $a$, $m$ and $n$ should satisfy (\ref{eq:mainc1}) and (\ref{eq:mainc2}). Equality $(b)$ is obtained by substituting $m=\frac{2a}{\left[1+0.5\epsilon\right]^{\frac{2}{n}}-1}$ from (\ref{eq:mainc2}), and by further assuming that $n\geq\left(\log\left(\frac{2+\epsilon}{\epsilon^2}\right)\right)^4$. Finally, (c) follows by noting that $\lim\limits_{n\rightarrow\infty}n\left(\left(1+x\right)^{\frac{2}{n}}-1\right)=2\log(1+x)$.
\end{proof}
\begin{remark}\normalfont
Proposition 2 reveals a fundamental communications limit of the considered wireless-powered system. To guarantee an $\epsilon$-reliable communication over $n$ channel uses, the node first needs to accumulate sufficient energy during the initial harvesting phase. A sufficiently large $m$ helps improve the energy availability at the transmitter. This harvesting overhead, however, causes a rate loss (versus a non-energy harvesting system) as the first $m$ channel uses are reserved for harvesting. Moreover, as the transmit blocklength $n$ grows, so does the length of the initial harvesting phase $m$, resulting in an inescapable performance limit on the communication system. This limit depends on i) the power ratio $a$, and ii) the required reliability $\epsilon$, and is captured by the prelog term $L(a,\epsilon)$ in (\ref{eq:prpo2a}) for a given $\gamma$. Moreover, this behavior is more visible for latency-constrained systems where the total blocklength is fixed. 
\end{remark}

\begin{remark}\normalfont
In the asymptotic blocklength regime, the harvest blocklength should be scaled proportionally to the transmit blocklength with a scaling rate $\frac{a}{\log(1+0.5\epsilon)}$ to attain the corresponding asymptotic $\epsilon$-achievable rate. Note that this scaling rate approximately equals $\frac{2a}{\epsilon}$ (when $\epsilon$ is small), which is similar to the non-asymptotic scaling rate discussed in Remark \ref{rem:scaling}.
\end{remark}
\begin{remark}\normalfont
We note that the asymptotic achievable rate vanishes as $\epsilon\rightarrow 0$. This is because the wireless energy transfer link may fade completely, resulting in a transmission outage for the information transfer link. 
\end{remark}
\begin{corollary}\normalfont
As the power ratio $a\rightarrow0$ in (\ref{eq:prop2}), the asymptotic achievable rate converges to the capacity of a non-energy harvesting AWGN channel, i.e., $\lim\limits_{a\rightarrow0}R^\infty_{\textrm{EH}}(\epsilon,a,\gamma)=C^\infty_{\textrm{AWGN}}(\gamma)$. 
\end{corollary}

\begin{remark}\normalfont
With $P_{\rm{t}}$ fixed, decreasing $a$ (by increasing $P_{\rm{E}}$) improves the energy availability at the EH node during the information transmission phase. As $a$ is decreased, a smaller harvest blocklength is required to support a certain transmit blocklength and $\epsilon$. As a result, in the limit $a\rightarrow 0$, the harvesting overhead vanishes as the transmit blocklength goes to infinity. Therefore, the system effectively reduces to a traditionally-powered communication system.     
\end{remark}

\begin{corollary}\normalfont
In the high-reliability regime (when $\epsilon\in[0,1)$ is small), the asymptotic achievable rate $R^\infty_{\textrm{EH}}(\epsilon,a,\gamma)$ in (\ref{eq:prop2}) can be approximated as 
\begin{align}
R^\infty_{\textrm{EH}}(\epsilon,a,\gamma)&\approx \frac{1}{1+\frac{2a}{\epsilon}}C^{\infty}_{\textrm{AWGN}}(\gamma)=\frac{1}{1+\frac{2P_t}{P_E\epsilon}}C^{\infty}_{\textrm{AWGN}}(\gamma),
\end{align}
which follows since $\log(1+x)\approx x$ when $x$ is small.
\end{corollary}

\begin{remark}\normalfont
The previous corollary illustrates an interesting interplay between the key design parameters. For a given target rate, the 
error probability $\epsilon$ scales inversely with the average harvested energy $P_{\rm{E}}$ in the high-reliability regime. This implies that increasing $P_{\rm{E}}$ (e.g., by increasing the PB transmit power) reduces the communication unreliability by the same factor.
\end{remark} 

\subsection{Optimal power control}
For optimal performance, the energy harvesting node needs to use the \emph{right} amount of transmit power. 
On the one hand, reducing $P_{\rm{t}}$ helps improve the energy supply probability as a packet transmission is less likely to face an energy outage. On the other hand, it is detrimental for the communication link as it reduces the SNR. We now quantify the optimal transmit power that maximizes the asymptotic achievable rate for a given set of parameters. We note that many of the analytical insights obtained for the asymptotic regime are also useful for the non-asymptotic regime (see Remark \ref{rem:power}).

\begin{corollary}\label{cor: opt power}\normalfont
For a given $\epsilon$ and $P_{\rm{E}}$, there exists an optimal transmit power that maximizes the achievable rate. We let $P_{\rm{t},\infty}^*$ be the rate-maximizing transmit power in the asymptotic blocklength regime. It follows that 
\begin{align}\label{eq: opt P_t}
P_{\rm{t},\infty}^*(\epsilon,P_{\rm{E}},\sigma^2)&=
\sigma^2\left(\frac{\frac{P_{\rm{E}}}{\sigma^2}\log(1+0.5\epsilon)-1}
{
\text{W}\left[\left(\frac{P_{\rm{E}}}{\sigma^2}\log(1+0.5\epsilon)-1\right) e^{-1}\right]
}-1\right)
\end{align}where $\text{W}[\cdot]$ is the Lambert W-function\cite{corless1996lambertw}. 
\begin{proof}
See Appendix A.
\end{proof}
\end{corollary}
Note that $\text{W}[x]$ is a real increasing function of $x$ for $x\geq -\frac{1}{e}$\cite{corless1996lambertw}. As $\frac{P_{\rm{E}}}{\sigma^2}\log\left(1+0.5\epsilon\right)> 0$ in practice, this ensures that the function $\text{W}\left[\left(\frac{P_{\rm{E}}}{\sigma^2}\log(1+0.5\epsilon)-1\right) e^{-1}\right]$ is real, resulting in a nonnegative transmit power. Also, 
plugging $P_{\rm{t}}=P_{\rm{t},\infty}^*$ in Proposition \ref{proposition:asym rate} gives the optimal achievable rate in the asymptotic blocklength regime.
Furthermore, when $P_{\rm{t}}$ is fixed, the achievable rate improves monotonically with $P_{\rm{E}}$ due to an increase in the energy supply probability. 

\begin{remark}\label{rem:power}\normalfont
The optimal transmit power for the asymptotic case serves as a conservative estimate for the optimal transmit power for the non-asymptotic case (Fig. \ref{fig:P_opt vs Pe}). Moreover, the achievable rate in the non-asymptotic regime obtained using the asymptotically optimal transmit power, gives a tight lower bound for the optimal achievable rate in the non-asymptotic regime (Fig. \ref{fig:rate vs. error}). This suggests that Corollary \ref{cor: opt power} provides a useful analytical handle for transmit power selection even for the finite blocklength regime (despite the fact that the resulting rate for the non-asymptotic case could be much smaller than that for the asymptotic case).  
\end{remark}

\begin{corollary}\label{cor:opt ratio}\normalfont
With $\epsilon$ and $\sigma^2$ fixed, the asymptotically optimal transmit power $P_{\rm{t},\infty}^*(\epsilon,P_{\rm{E}},\sigma^2)$ increases with $P_{\rm{E}}$ with a slope 
\begin{align}
\frac{\log(1+0.5\epsilon)}{1+\text{W}\left[\left(\frac{P_{\rm{E}}}{\sigma^2}\log(1+0.5\epsilon)-1\right)e^{-1}\right]}.
\end{align}
The slope is a non-negative decreasing function of the $P_{\rm{E}}$, suggesting that i) the optimal transmit power increases monotonically with $P_{\rm{E}}$, and ii) it is more sensitive to $P_{\rm{E}}$ when $P_{\rm{E}}$ is small. In addition, the optimal transmit power scales sublinearly with $P_{\rm{E}}$.
\begin{proof}
It follows by differentiating the optimal transmit power with respect to $P_{\rm{E}}$.
\end{proof}
\end{corollary}
Though the transmit power increases with $P_{\rm{E}}$, the optimal power ratio $a^*=\frac{P_{\rm{t,\infty}}^*}{P_{\rm{E}}}$ is a monotonically decreasing function of $P_{\rm{E}}$. This is because $P_{\rm{t,\infty}}^*$ varies sublinearly with $P_{\rm{E}}$. 

\section{Multiple Power Beacons}\label{secMP}
In this section, we extend the analysis to the case of a large-scale network consisting of power beacons, wireless-powered transmitters, and their dedicated receivers. 
We assume that the power beacons are distributed on a two-dimensional plane according to a homogeneous Poisson point process (PPP) $\Phi=\{x_k\}_{k=1}^{\infty}$ with density (intensity) $\lambda$, where $x_k$  denotes the location of a node $k$ in $\Phi$. 
The energy harvesting transmitters are drawn from another homogeneous PPP independently of the power beacons.
Similar to the case of a single power beacon, each energy harvesting transmitter is assumed to have a dedicated receiver.
Leveraging Slivnyak's theorem\cite{haenggi2012stochastic}, we consider a typical energy harvesting node located at the origin. 
It exploits the energy harvested from the transmissions of multiple power beacons to communicate with its dedicated receiver over a noisy channel. This implicitly assumes that an EH transmitter causes negligible interference to other EH-RX links, since the transmit power of an EH node is usually very small. 
We let $h_k$ model the small-scale fading coefficient for the PB-EH link originating at $x_k$. We assume IID Rayleigh fading for the PB-EH links such that $H_k=|h_k|^2\sim\exp(1)$. As defined previously, $\ell\left(\|x_k\|,\eta\right)$ models the distance-dependent path loss for the link from $x_k$. The energy harvested in an arbitrary channel use for the case of multiple power beacons is given by $Z^{}=P_{\rm{PB}}\mu\sum\limits_{x_k\in\Phi}{\frac{H_k}{ \ell(\|x_k\|, \eta)}}$.
We derive tractable analytical expressions for the energy supply probability and the non-asymptotic achievable rate in a network setting.
\subsection{Energy Supply Probability}
We first characterize the energy supply probability in a general form. We then specialize it to the scenario considered in this paper. 
\begin{proposition}\normalfont\label{lemma: multi_pec}
For the case of multiple power beacons with PB density $\lambda$, the energy supply probability at a typical EH node is given by 
	\begin{align}
		P_{\textrm{es}}^{\rm{MP}}\left(m,n,a,\lambda,\eta\right)&=
1-\sum\limits_{i=0}^{\frac{n}{2}-1}(-1)^{i}\frac{{m}^{i}}{(2a)^ii!}\frac{\text{d}^i}{\text{d}s^i}\mathcal{L}_Z(s)|_{s=\frac{m}{2a}}
	\end{align}
	where the power ratio $a=\frac{P_{\rm{t}}}{\mu P_{\rm{PB}}}$, $\eta$ is the path loss exponent, while $\mathcal{L}_Z(s)=\mathbb{E}[e^{-sZ}]$ is the Laplace transform of the per-slot harvested energy $Z$, which is also a function of $\lambda$ and $\eta$. 
	\begin{proof}
	See Appendix B.
	\end{proof}
\end{proposition}
Note that the power ratio $a$ is defined here slightly differently from the case of single power beacon (Proposition 1). Here, it is defined as the ratio of the transmit power at an energy harvester to that at a power beacon. Previously, it was defined as the ratio of the EH transmit power to the harvested power, i.e., the large-scale fading term, being deterministic, was absorbed in the power ratio. 
For generality, we have expressed Proposition \ref{lemma: multi_pec} in terms of the Laplace transform of the harvested energy. Depending on the propagation and network model, this could be evaluated in closed form. For example, the following lemma analytically characterizes the Laplace transform for the scenario relevant to this paper.
\begin{lemma}\normalfont\label{lemma:laplace}
Let us assume the PBs are drawn from a homogeneous PPP of density $\lambda$, the PB-EH links are IID Rayleigh fading, and follow a bounded path loss model $\ell(r,\eta)=\max(1,r^\eta)$ where $\eta>2$ is the path loss exponent while $r$ is the PB-EH link distance. The Laplace transform $\mathcal{L}_{Z}(s)$ of the per-slot harvested energy $Z$ is analytically characterized by 
\begin{align}\label{eq:laplace}
\mathcal{L}_{Z}(s)=\exp\left(-\pi\lambda\frac{P_{\rm{PB}}\mu s}{1+P_{\rm{PB}}\mu s}\right)
\exp\left(-\pi\lambda~\mathcal{F}\left(P_{\rm{PB}}\mu s,\eta\right)\right),
\end{align}
where the function $\mathcal{F}\left(x_1,x_2\right)$ for $x_1\geq 0, x_2> 2$ is defined as  
\begin{align}\label{eq:functionF}
\mathcal{F}\left(x_1,x_2\right)=\frac{2x_1}{x_2-2}~{_{2}{F}_{1}}\left(1,1-\frac{2}{x_2};2-\frac{2}{x_2};-x_1\right)
\end{align}
in terms of the Gauss's hypergeometric function ${_{2}{F}_{1}}\left(c_1,c_2\,;c_3\,;z\right)$ \cite{prudnikov1998integrals}. 
\begin{proof}
See Appendix B
\end{proof}
\end{lemma}
We note that the Laplace transform is expressed in terms of tractable mathematical functions, which can be evaluated using most numerical toolboxes. We now characterize the mean harvested energy in terms of the network density and the path loss exponent.
\begin{lemma}\label{lemm:mean}\normalfont
The average per-slot harvested energy for the case of multiple power beacons is given by $\mathbb{E}\left[Z\right]=\lambda\pi \frac{\eta}{\eta-2}\mu P_{\rm{PB}}.$ This shows that the $\lambda$ and $P_{\rm{PB}}$ have the same effect on the mean harvested energy. 
\begin{proof}
See Appendix B. 
\end{proof}
\end{lemma}
The following lemmas treat the partial derivatives of the functions involved in the Laplace transform. We will apply them in the analytical characterization of the energy supply probability for the propagation model considered in this paper. 
\begin{lemma}\label{lemma:derivative2}\normalfont
We let ${_{2}{F}_{1}^{(k)}}\left(1,1-\frac{2}{x_2};2-\frac{2}{x_2};-x_1\right)$ denote the $k$th-order partial derivative of the function ${_{2}{F}_{1}^{}}\left(1,1-\frac{2}{x_2};2-\frac{2}{x_2};-x_1\right)$ with respect to the variable $x_1$, where $k=0$ refers to the original function. 
Using the properties of the hypergeometric function\cite{prudnikov1998integrals}, it follows that 
\begin{align}\label{eq:2F1n}
{_{2}{F}_{1}^{(k)}}\left(1,1-\frac{2}{x_2};2-\frac{2}{x_2};-x_1\right)
&=\nonumber\\
(-1)^k k!&\frac{\left(1-\frac{2}{x_2}\right)_{(k)}}{\left(2-\frac{2}{x_2}\right)_{(k)}}~
{_{2}{F}_{1}^{(0)}}\left(k+1,k+1-\frac{2}{x_2};k+2-\frac{2}{x_2};-x_1\right)
\end{align}
where $\left(x\right)_{(k)}=\frac{\Gamma(x+k)}{\Gamma(x)}$ is the Pochhammer symbol, while $\Gamma(x)=\int\limits_{0}^{\infty}t^{x-1}e^{-t}\text{d}t$ is the Gamma function\cite{prudnikov1998integrals}.
\end{lemma}

\begin{lemma}\label{lemma:derivative1}\normalfont
We let $\mathcal{F}^{(k)}(x_1,x_2)$ denote the $k$th order partial derivative of the function 
$\mathcal{F}^{}(x_1,x_2)$ with respect to the variable $x_1$. It follows that 
\begin{align}\label{eq:Fn}
\mathcal{F}^{(k)}\left(x_1,x_2\right)=
\frac{2k}{x_2-2}~{_{2}{F}_{1}^{(k-1)}}&\left(1,1-\frac{2}{x_2},2-\frac{2}{x_2},-x_1\right)\nonumber\\
&+
\frac{2x_1}{x_2-2}~{_{2}{F}_{1}^{(k)}}\left(1,1-\frac{2}{x_2},2-\frac{2}{x_2},-x_1\right)
\end{align}
where $\mathcal{F}^{(0)}\left(x_1,x_2\right)=\mathcal{F}^{}\left(x_1,x_2\right)$.
\begin{proof}
The result follows by successive differentiation of (\ref{eq:functionF}) with respect to $x_1$, invoking Lemma \ref{lemma:derivative2}, and (recursively) expressing the result in terms of the lower-order derivatives of the original function.
\end{proof}
\end{lemma}
Leveraging Lemma \ref{lemma:laplace} and Fa\`{a} di Bruno formula\cite{noschese2003differentiation}, we now specialize Proposition \ref{lemma: multi_pec} to the scenario considered in this paper.  	
\begin{proposition}\normalfont\label{prop:mp_exact}
The energy supply probability for the bounded path loss model considered in Lemma \ref{lemma:laplace} can be expressed in closed-form as
\begin{align}
P^{\rm{MP}}_{\textrm{es}}\left(m,n,a,\lambda,\eta\right)=
e^{-\pi\lambda\left(\frac{s}{1+s}+\mathcal{F}\left(s,\eta\right)\right)}
\sum\limits_{i=0}^{\frac{n}{2}-1}\frac{(-s)^{i}}{i!}~
B_i\left(g^{(1)}(s),\cdots,g^{(i)}(s)\right)\bigg|_{s=\frac{m}{2a}}
\end{align}
where $B_i(u_1,\cdots,u_i)$ is the complete Bell polynomial of the second kind \cite{noschese2003differentiation}, and 
\begin{align}
g^{(i)}(s)=-\pi\lambda\left(\left[-\frac{1}{1+s}\right]^{i+1}i!+
~\frac{i}{s}\Upsilon(i-1,\eta)\mathcal{F}^{(i-1)}(s,\eta)
+~\Upsilon(i,\eta)\mathcal{F}^{(i)}(s,\eta)
\right),
\end{align}
where 
$\mathcal{F}^{(i)}\left(x_1,x_2\right)$ is given in Lemma \ref{lemma:derivative1} and $\Upsilon\left(i,x_2\right)=(-1)^{i}\,i!\,\frac{\left(1-\frac{2}{x_2}\right)_{\left(i\right)}}{\left(2-\frac{2}{x_2}\right)_{\left(i\right)}}$.
\begin{proof}
The proof follows by invoking Fa\`{a} di Bruno formula\cite{noschese2003differentiation} to calculate the partial derivatives of the Laplace transform in Lemma \ref{lemma:laplace}, and applying Lemma \ref{lemma:derivative2} and \ref{lemma:derivative1}.
\end{proof}
\end{proposition}
We note that the energy supply probability in Proposition \ref{prop:mp_exact} is expressed in terms of numerically tractable mathematical functions, which can be evaluated using most numerical toolboxes. Moreover, our analytical treatment is fairly general since Proposition 3 can be specialized to various scenarios, similar to the derivation of Proposition 4.

\subsection{Achievable Rate}
Leveraging the results in the previous sections, we now provide an analytical treatment of the achievable rate for the case of multiple power beacons.
\begin{theorem}\normalfont\label{th:MP}
	When the EH nodes are powered by multiple PBs distributed with a density $\lambda$, the non-asymptotic $\epsilon$-achievable rate at a typical intended receiver is characterized by 
	\begin{align}\label{eq:maind0}
	R_{\rm{EH}}^{\rm{MP}}\left(\epsilon,a,\gamma,m,n,\lambda\right)=\frac{\frac{n\log(1+\gamma)}{2}
		-\sqrt{\frac{2+\epsilon}{\epsilon}\frac{\gamma}{\gamma+1}n}
		-{(n)}^{\frac{1}{4}}-1}
	{n+m}
	\end{align}
	for all tuples $(m,n)$ satisfying the following constraints.
	\begin{align}\label{eq:maind1}
	\sum\limits_{i=0}^{\frac{n}{2}-1}(-1)^{i}\frac{{m}^{i}}{(2a)^ii!}\frac{\text{d}^i}{\text{d}s^i}\mathcal{L}_Z(s)\bigg|_{s=\frac{m}{2a}} \leq\frac{\epsilon}{2+\epsilon}
	\end{align}
	where $\mathcal{L}_{Z}\left(s\right)$ follows from Lemma \ref{lemma:laplace}; and
	\begin{align}\label{eq:maind2}
	n\geq \left[\log\left(\frac{2+\epsilon}{\epsilon^2}\right)\right]^4 .
	\end{align}	
\end{theorem}
\begin{proof}
	See Appendix B.
\end{proof}
The achievable rate expression for the case of multiple power beacons can be interpreted similar to the case of a single power beacon. For example, we may evaluate the expression following the minimum latency approach defined previously. For generality, we have expressed
Theorem \ref{th:MP} in terms of the Laplace transform, which can be evaluated using Proposition \ref{prop:mp_exact}. 

\section{Numerical Results}\label{secSim}
We now present the simulation results for the energy supply probability and the achievable rate based on the analyses in Section \ref{secAnl} and \ref{secMP}.
We assume that the noise power $\sigma^2=1$, the rectifier efficiency $\mu=1$, and path loss exponent $\eta=3.6$. We do not specify the units of $P_{\rm{t}}$, $P_{\rm{PB}}$, or $P_{\rm{E}}$ since the results are valid for any choice of the units (say Joules/symbol).
\subsection{Single Power Beacon}
We first present the results for the case of a single power beacon treated in Section \ref{secAnl}.
In the following plots, we adopt the minimum latency approach where the minimum possible blocklength is selected for the given set of parameters, based on the constraints in (\ref{eq:mainc1}) and (\ref{eq:mainc2}). That is, for a given $\epsilon$, we select the minimum required $n$ using $n= \Big\lceil\left(\log\left(\frac{2+\epsilon}{\epsilon^2}\right)\right)^4\Big\rceil_{\rm{ev}}$. We then choose the minimum required $m$ using (\ref{eq:mainc22}). 
In Fig. \ref{fig:rate vs ratio}, we use Theorem 1 and Proposition 2 to plot the achievable rate versus the power ratio $a$ for a given $\epsilon$ and $P_{\rm{E}}$. The plot reflects the underlying tension between the energy supply probability and the channel SNR, resulting in an optimal transmit power (or power ratio) that maximizes the achievable rate. 
We also observe that the EH node can transmit at a higher rate as the target error probability is increased.

In Fig. \ref{fig:rate vs. error}, we plot the achievable rate versus the target error probability $\epsilon$ for a  given power ratio $a$. 
 We first consider the (fixed power) case where we fix the transmit power $P_{\rm{t}}=1.1554$ and the power ratio $a=0.0012$ (these values are asymptotically optimal for $P_{\rm{E}}=10^3$ and $\epsilon=10^{-3}$).
 As $\epsilon$ increases, the achievable rate tends to increase until a limit, beyond which the rate tends to decrease. This is because as we allow for more error ($\epsilon \uparrow$), the required total blocklength decreases. This means a possible increase in the energy supply probability (as the power ratio is fixed), and a larger backoff from capacity due to a shorter transmit blocklength. Beyond a certain $\epsilon$, further reduction in blocklength pronounces the higher order backoff terms, eventually reducing the rate. For a \emph{fixed} total blocklength, however, the achievable rate indeed increases with $\epsilon$. We note that these trends differ from the asymptotic case where the rate monotonically increases with $\epsilon$. 
  We then consider the case where we adapt the transmit power using Corollary \ref{cor: opt power}. In Fig. \ref{fig:rate vs. error}, we observe a substantial increase in the rate by optimally adjusting the transmit power in terms of the system parameters. Moreover, using the \emph{asymptotically} optimal transmit power $P^*_{\rm{t,\infty}}$ (from Corollary \ref{cor: opt power}) in the finite blocklength regime results in only a minor loss in performance. As evident from Fig. \ref{fig:rate vs. error}, the optimal rate in the finite blocklength regime (obtained by numerically optimizing over $P_{\rm{t}}$) is almost indistinguishable from the lower bound obtained using the asymptotically optimal power $P^*_{\rm{t,\infty}}$.

\begin{figure} [t]
	\centerline{
		\includegraphics[width=.7\columnwidth]{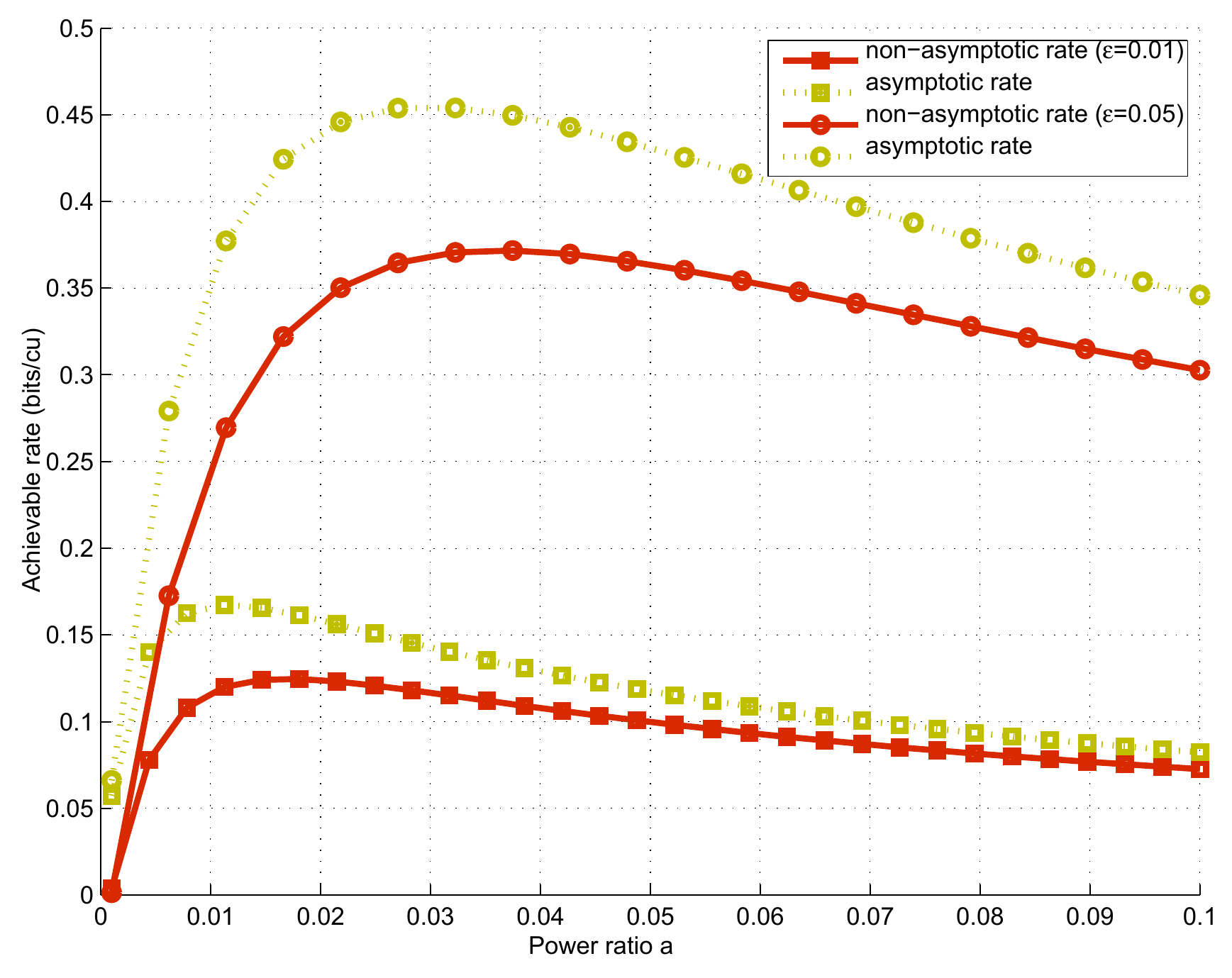}}
	\caption{The achievable rate (bits/channel use) vs. the power ratio $a=\frac{P_{\rm{t}}}{P_{\rm{E}}}$ for $P_{\rm{E}}={10}^2$. There is an optimal transmit power that maximizes the rate.}
	\label{fig:rate vs ratio}
\end{figure}

\begin{figure} [t]
	\centerline{
		\includegraphics[width=0.7\columnwidth]{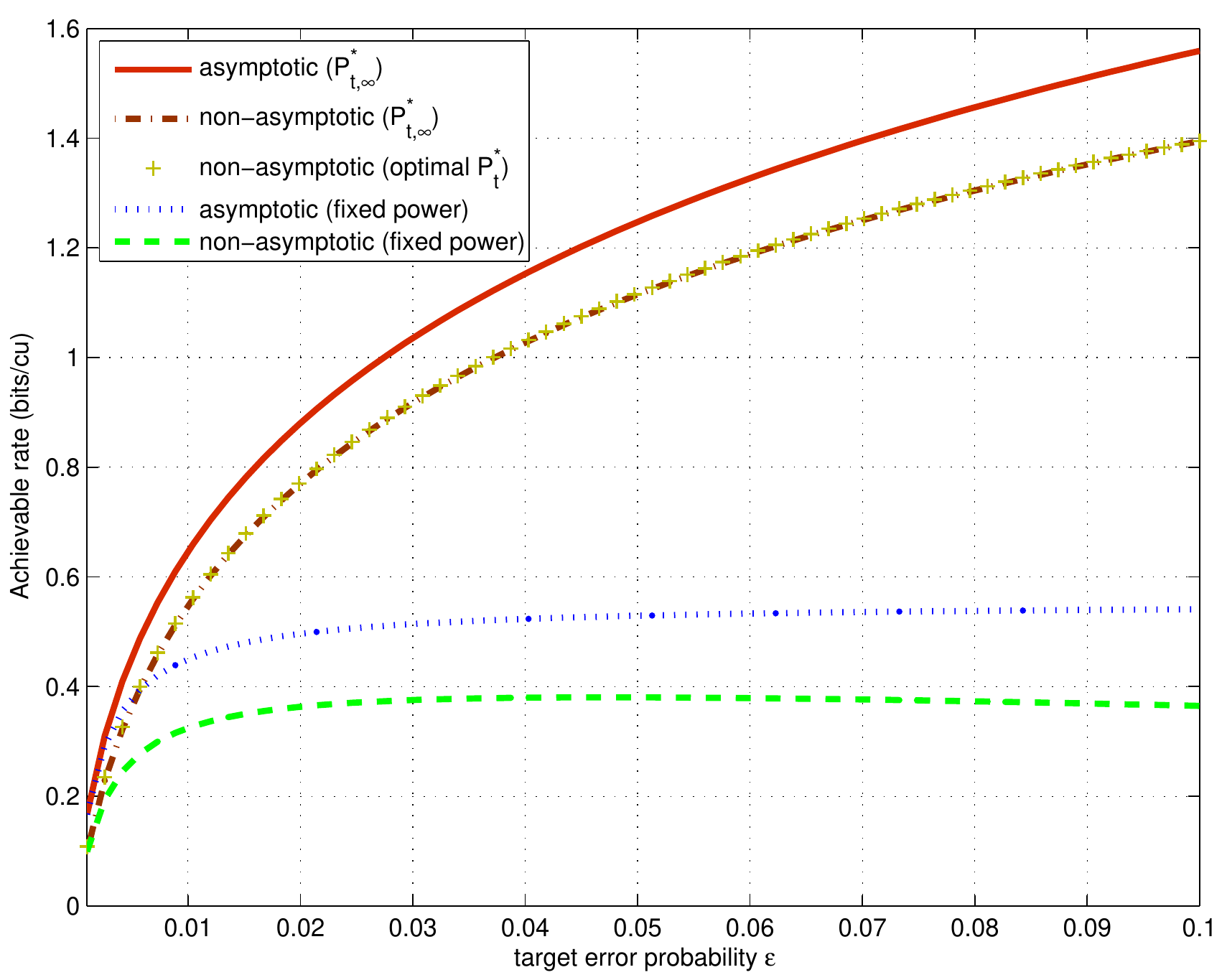}}
	\caption{The achievable rate (bits/channel use) vs. the target error probability $\epsilon$ for a given power ratio $a=0.0012$. 
	While the asymptotic rate increases as we allow for more error, the non-asymptotic rate behaves differently. Moreover, power control is essential for improving the achievable rate.}
	\label{fig:rate vs. error}
\end{figure}

\begin{figure} [t]
	\centerline{	\includegraphics[width=.7\columnwidth]{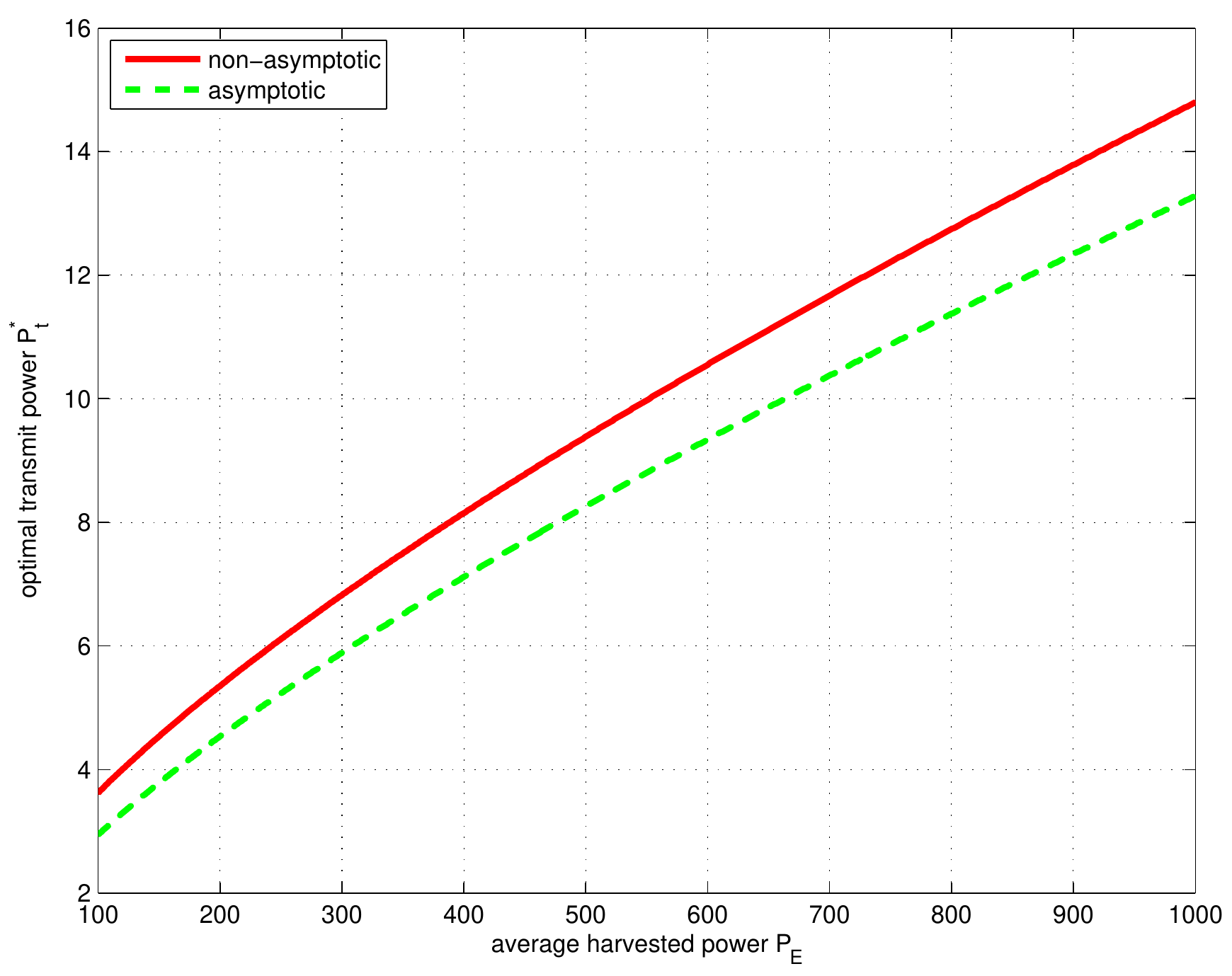}}
	\caption{Optimal transmit power $P_{\rm{t}}$ vs. average harvested power $P_{\rm{E}}$ in the asymptotic and non-asymptotic blocklength regimes. The asymptotically optimal transmit power is a conservative estimate of the non-asymptotic transmit power.}
	\label{fig:P_opt vs Pe}
\end{figure}

\begin{figure} [t]
	\centerline{	\includegraphics[width=.7\columnwidth]{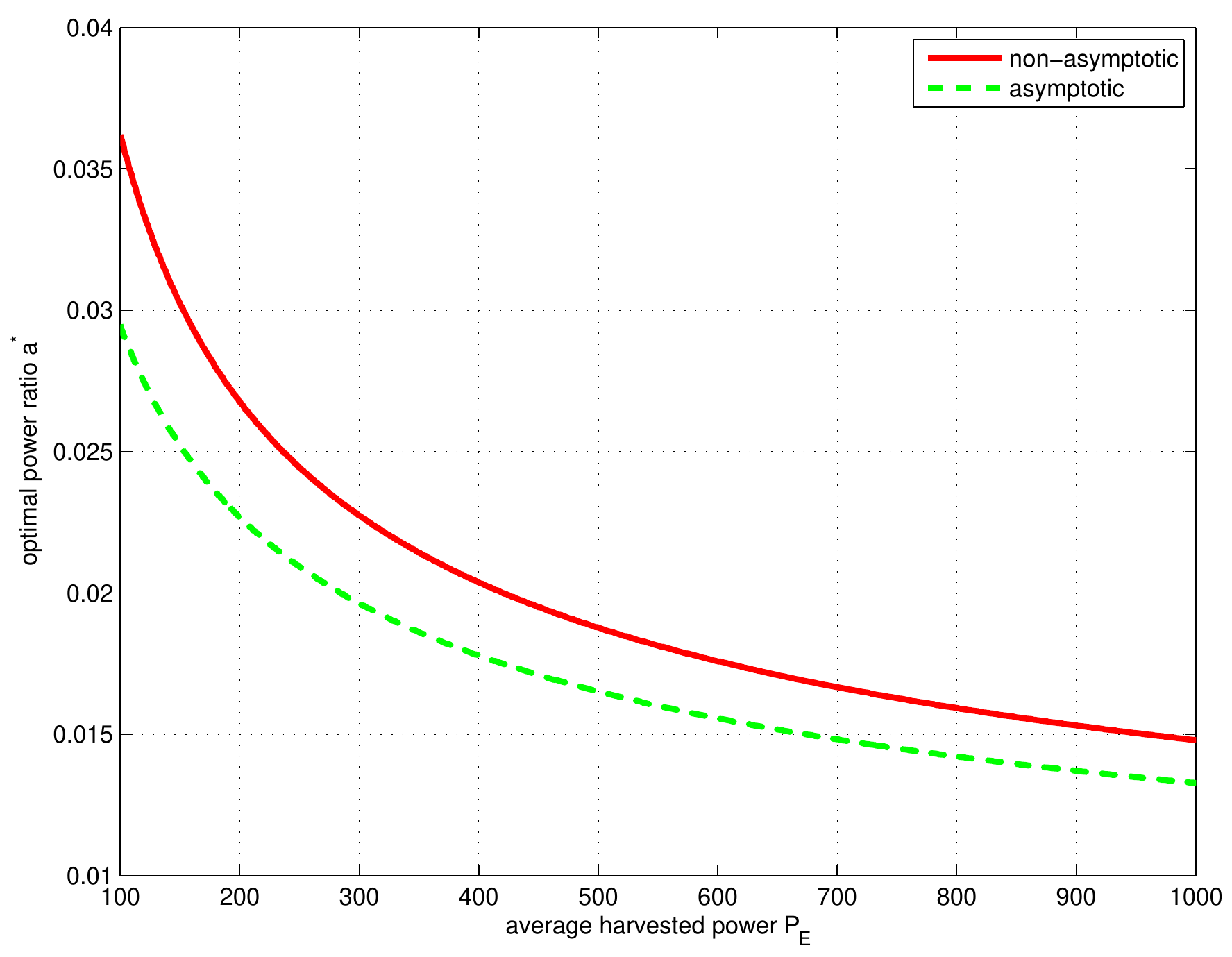}}
	\caption{Optimal power ratio $a^*$ vs. average harvested power $P_{\rm{E}}$ in asymptotic and non-asymptotic blocklength regimes. The optimal power ratio decays as the average harvested power is increased.}
	\label{fig:a_opt vs Pe}
\end{figure}

In Fig. \ref{fig:P_opt vs Pe}, we plot the optimal transmit power versus the average harvested power for $\epsilon=0.05$ and the transmit blocklength $n=\lceil \log\left(\frac{2+\epsilon}{\epsilon^2}\right)^4\rceil_{\rm{ev}}=2026$. For each $P_{\rm{E}}$, the harvest blocklength is selected to satisfy the constraints in (\ref{eq:mainc1}) and (\ref{eq:mainc2}). We observe that the asymptotically optimal transmit power is a conservative estimate of the optimal transmit power for the finite case (Remark \ref{rem:power}). 
In Fig. \ref{fig:a_opt vs Pe}, we plot the optimal power ratio against the average harvested power.
Even though the optimal transmit power increases with $P_{\rm{E}}$,  
we note that the optimal power ratio still decreases as $P_{\rm{E}}$ is increased. In other words, while it is optimal to increase $P_{\rm{t}}$ with $P_{\rm{E}}$, the scaling is sublinear in $P_{\rm{E}}$ (Corollary \ref{cor:opt ratio}). 
 
\begin{figure} [t]
	\centerline{	\includegraphics[width=.7\columnwidth]{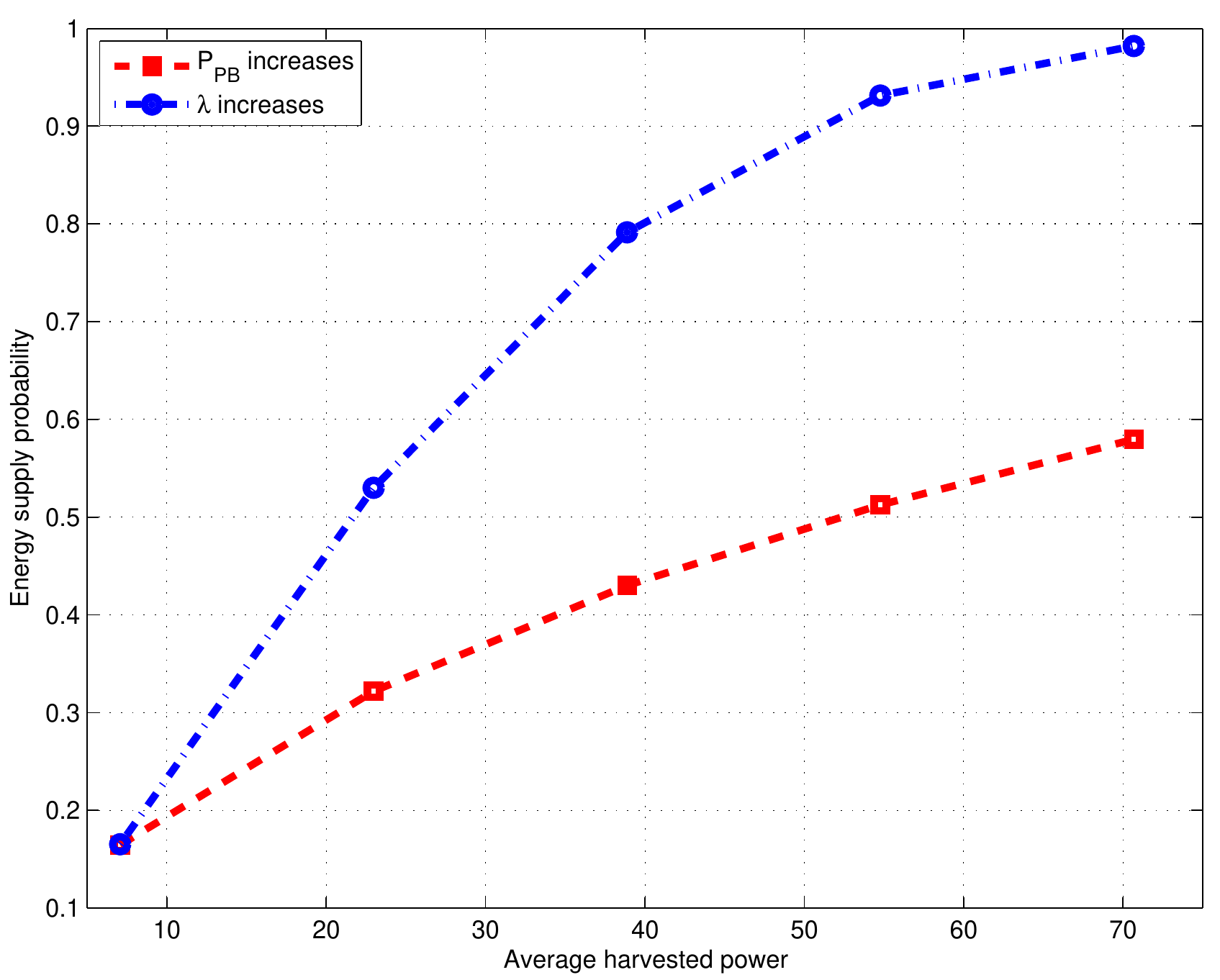}}
	\caption{The energy supply probability $P_{\rm{es}}(m,n,a,\lambda,\eta)$ vs. the average harvested power for $m=1500$, $n=1000$, $P_{\rm{t}}=1$. For the same \textit{mean} harvested power, increasing the PB density is more beneficial than increasing the PB transmit power.}
	\label{fig:Pes vs density}
\end{figure}
\subsection{Multiple Power Beacons}
We now consider the case of multiple power beacons treated in Section \ref{secMP}. 
In Fig \ref{fig:Pes vs density}, we plot the energy supply probability versus the mean harvested power for a fixed total blocklength and the EH transmit power. 
The average harvested power is increased by increasing either the PB transmit power $P_{\rm{PB}}$ or the PB density $\lambda$, according to Lemma \ref{lemm:mean}.
We consider two cases: i) $\lambda$ is fixed and $P_{\rm{PB}}$ is increased, and ii) $P_{\rm{PB}}$ is fixed and $\lambda$ is increased. For the former, we obtain the plot for $P_{\rm{PB}}$ ranging from $10^{3}$ to $10^{4}$ and $\lambda=10^{-3}$ nodes per $\rm{m}^2$.  
For the latter, we assume $\lambda$ ranges from $10^{-3}$ to $10^{-2}$ nodes per $\rm{m}^2$ and $P_{\rm{PB}}=10^3$.  
Keeping the average harvested power same in both cases, we observe that increasing the PB density is more beneficial for the energy supply probability than increasing the PB transmit power.
Finally, in Fig. \ref{fig: Rate vs blklength}, we invoke Theorem \ref{th:MP} to plot the achievable rate versus the transmit blocklength under the minimum latency approach. Moreover, we numerically optimize over the transmit power $P_{\rm{t}}$ for each $n$. 
We observe that the achievable rate is extremely sensitive to the blocklength, confirming that the asymptotic analyses fail to capture the behavior of a wirelessly powered system with short packets.

\begin{figure} [t]
	\centerline{	\includegraphics[width=.7\columnwidth]{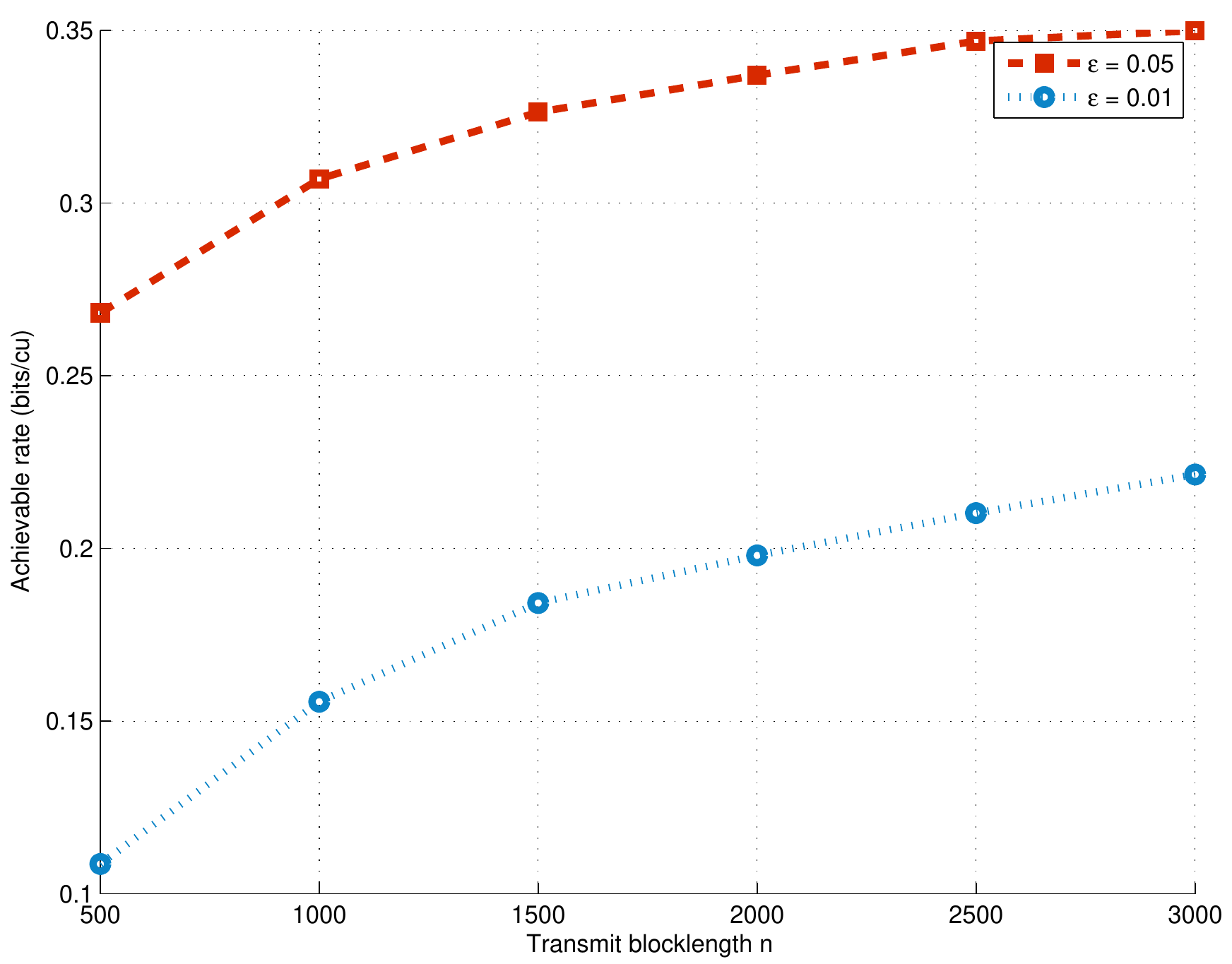}}
	\caption{The $\epsilon$-achievable rate versus transmit blocklength at a typical harvester powered by multiple power beacons ($\lambda=0.005$ nodes per $\rm{m}^2$). The achievable rate improves as the blocklength is increased, confirming that the non-asymptotic rate is substantially smaller than the asymptotic rate.}
	\label{fig: Rate vs blklength}
\end{figure}

\section{Conclusions}\label{secConc}
We characterized the energy supply probability and the achievable rate of a wireless-powered communication system in the finite blocklength regime. Using analytical expressions as well as numerical simulations, we investigated the interplay between key system parameters such as the harvest blocklength, the transmit blocklength, the error probability, and the power ratio. 
For the case of a single power beacon, we showed that the harvest blocklength should be scaled proportionally to the transmit blocklength in order to maintain the $\epsilon$-achievable rate. The rate of growth is characterized by the power ratio as well as the target error probability.
Moreover, we derived closed-form expression for the optimal transmit power in the asymptotic blocklength regime. 
Numerical results show that using the asymptotically optimal transmit power can substantially improve the achievable rate even in the finite blocklength regime. We also extended the analysis to a large-scale network with Poisson-distributed power beacons. Numerical results reveal that the performance is sensitive to the blocklength, confirming that the asymptotic analyses of wireless-powered systems fail to capture the behavior in the short packet regime.
\appendices
\allowdisplaybreaks
\section*{Appendix A: Single Power Beacon}
\subsection*{Proof of Theorem 1}
The proof leverages the fact that the communication link failure mainly results from two events: energy outages at the transmitter or decoding error at the receiver. The first step of the proof involves bounding the decoding errors due to energy outages and channel noise in terms of the target error probability. The second step uses conventional information theoretic arguments to derive an expression for the non-asymptotic achievable rate for the considered wireless-powered channel. 
Let us first bound the energy outage probability as
\begin{align}\label{eq:eneergyBound1}
\Pr\left[{\bigcup_{k=1}^{n}}\left\{\sum_{\ell=1}^{k}X_\ell^2\geq\sum_{i=1}^{m}Z_i\right\} \right]&\leq 1-\frac{2}{2+\epsilon}
\end{align}
for $\epsilon\in[0,1)$. Using Lemma \ref{lemma:energy}, the constraint in (\ref{eq:eneergyBound1}) can be equivalently expressed in terms of the energy supply probability as
$\Pr\left[\sum_{\ell=1}^{n}X_\ell^2\leq\sum_{i=1}^{m}Z_i \right]
\geq \frac{2}{2+\epsilon}$. 
We let $X^n(W)$ and $Y^n$ denote the intended codeword sequence for a message $W\in\mathcal{W}$, and the received sequence. The following proof is inspired by the proof techniques in \cite{fong2015non}. The decoder $\mathcal{G}({Y}^n)$ employs the following threshold decoding rule \cite{fong2015non} to decode the received signal: $\mathcal{G}({Y}^n)=i$ if there exists a unique integer $i\in\mathcal{W}$ that satisfies
\begin{align}\label{eq:rule}
\log\left(\frac{p_{Y^n|X^n}\left(Y^n|X^n(i)\right)}{p_{Y^n}\left(Y^n\right)}\right)>\log(M)+n^{\frac{1}{4}},
\end{align}
otherwise $\mathcal{G}({Y}^n)=w$, where $w$ is drawn uniformly at random from $\mathcal{W}$. 
Here, the notation $p_{Y^n|X^n}(\cdot)$ denotes the joint conditional distribution of random sequence $Y^n$ given $X^n$.
We express the probability of decoding error
$\Pr\left[\mathcal{G}({Y}^n)\neq W\right]$ in (\ref{eq:decodingerror}).
\begin{align}\label{eq:decodingerror}
\Pr&\left[\mathcal{G}({Y}^n)\neq W\right]=\nonumber\\
&\Pr\left[\mathcal{G}({Y}^n)\neq W, {Y}^n=X^n(W)+V^n\right]
+\Pr\left[\mathcal{G}({Y}^n)\neq W, {Y}^n\neq X^n(W)+V^n\right]\nonumber\\
&\leq\Pr\left[\mathcal{G}(X^n(W)+V^n)\neq W\right]+
\frac{\epsilon}{2+\epsilon},
\end{align}
where the inequality results from (\ref{eq:eneergyBound1}).
To calculate $\Pr\left[\mathcal{G}(X^n(W)+V^n)\neq W\right]$, we define ${\mathcal{A}}_{i|j}$ as the event that $i\in\mathcal{W}$ satisfies the threshold decoding rule of (\ref{eq:rule}) when $j\in\mathcal{W}$ is transmitted, i.e., 
\begin{align}\label{eq:Aij}
{\mathcal{A}}_{i|j}=\left\{\log\left(\frac{p_{Y^n|X^n}\left(X^n(j)+V^n|X^n(i)\right)}{p_{Y^n}\left(X^n(j)+V^n\right)}\right)>\log(M)+n^{\frac{1}{4}}\right\},
\end{align}
and ${\mathcal{A}}_{i|j}^c$ denotes its complement. As the message $W$ is uniform on $\mathcal{W}$, it follows that the decoding error probability 

\allowdisplaybreaks\begin{align}
\Pr\left[\mathcal{G}(X^n(W)+V^n)\neq W\right]
&\overset{(a)}{=}\frac{1}{M}\sum_{w=1}^{M}\Pr\left[{\mathcal{A}}_{w|w}^c\bigcup\bigcup_{i\neq w, i\in\mathcal{W}}{\mathcal{A}}_{i|w}\Big|W=w\right]\nonumber\\
&\overset{(b)}{=}\Pr\left[{\mathcal{A}}_{1|1}^c\bigcup\bigcup_{i=2}^M{\mathcal{A}}_{i|1}\right]
\overset{(c)}{\leq}\Pr\left[{\mathcal{A}}_{1|1}^c\right]+\Pr\left[\bigcup_{i=2}^M{\mathcal{A}}_{i|1}\right]\nonumber\\
&\overset{(d)}{\leq}\Pr\left[{\mathcal{A}}_{1|1}^c\right]+e^{-n\delta}
\overset{(e)}{\leq}\Pr\left[{\mathcal{A}}_{1|1}^c\right]+\frac{\epsilon^2}{2+\epsilon}
\end{align}
where (b) follows from the symmetry in random codebook construction, (c) results from applying the Union bound, and (d) is obtained by invoking Lemma 3 from \cite{fong2015non}. Finally, (e) follows by setting $n\delta=n^{\frac{1}{4}}$, and by further noting that 
$n\geq \left(\log\left(\frac{2+\epsilon}{\epsilon^2}\right)\right)^4$,
which follows from the constraint in (\ref{eq:mainc11}).
Before proceeding further, 
let us assume that $M$ is a unique integer that satisfies (\ref{eq:assumeM}).

\begin{align}\label{eq:assumeM}
\log(M+1)\geq n\mathbb{E}\left[\log\left(\frac{p_{Y|X}(Y|X))}{p_{Y}(Y)}\right)\right]
-\left({\frac{2+\epsilon}{\epsilon}n\mathtt{{Var}}\left[\log\left(\frac{p_{Y|X}(Y|X))}{p_{Y}(Y)}\right)\right]}\right)^{\frac{1}{2}}>\log(M)
\end{align}
To find a bound for $\Pr\left[{\mathcal{A}}_{1|1}^c\right]$, consider the following set of inequalities in (\ref{eq:q1}) 
\begin{align}
\label{eq:q1}
\Pr\left[{\mathcal{A}}_{1|1}^c\right]&\overset{(a)}{=}\Pr\left[\log\left(\frac{p_{Y^n|X^n}(X^n(1)+V^n|X^n(1))}{p_{Y^n}(X^n(1)+V^n)}\right)\leq\log(M)+n^{\frac{1}{4}}\right]\nonumber\\
&=\Pr\left[\sum_{k=1}^{n}\log\left(\frac{p_{Y|X}(X_k(1)+V_k|X_k(1))}{p_{Y}(X_k(1)+V_k)}\right)\leq\log(M)+n^{\frac{1}{4}}\right]\nonumber\\
&\overset{(b)}{\leq}\Pr\Bigg[\sum_{k=1}^{n}\log\left(\frac{p_{Y|X}(X_k(1)+V_k|X_k(1))}{p_{Y}(X_k(1)+V_k)}\right)\leq\nonumber\\
&\qquad\qquad\qquad n\mathbb{E}\left[\log\left(\frac{p_{Y|X}(Y|X))}{p_{Y}(Y)}\right)\right]-\left(\frac{2+\epsilon}{\epsilon}n\textrm{Var}\left[{\log\left(\frac{p_{Y|X}(Y|X))}{p_{Y}(Y)}\right)}\right]\right)^{\frac{1}{2}}\Bigg]\nonumber\\
&\leq\Pr\Bigg[\left|\sum_{k=1}^{n}\log\left(\frac{p_{Y|X}(X_k(1)+V_k|X_k(1))}{p_{Y}(X_k(1)+V_k)}\right)-
n\mathbb{E}\left[\log\left(\frac{p_{Y|X}(Y|X))}{p_{Y}(Y)}\right)\right]
\right|\geq\nonumber\\
&\qquad\qquad\qquad\qquad\qquad\qquad\left(\frac{2+\epsilon}{\epsilon}n\textrm{Var}\left[{\log\left(\frac{p_{Y|X}(Y|X))}{p_{Y}(Y)}\right)}\right]\right)^{\frac{1}{2}}\Bigg]\nonumber\\
&\overset{(c)}{\leq}\frac{\epsilon}{2+\epsilon}
\end{align}
where ($a$) follows from the definition of ${\mathcal{A}}_{i|j}$ in (\ref{eq:Aij}), while the bound in ($b$) results from (\ref{eq:assumeM}). Finally, ($c$) is obtained by applying Chebychev's inequality. From (\ref{eq:Aij}) and (\ref{eq:q1}), it follows that 
$\Pr\left[\mathcal{G}(X^n(W)+V^n)\neq W\right]
=\frac{\epsilon+\epsilon^2}{2+\epsilon}$;
and further using (\ref{eq:decodingerror}), we conclude that 
$\Pr\left[\mathcal{G}({Y}^n)\neq W\right]\leq\epsilon$,
where $W$ is the transmitted message. Therefore, we conclude that the constructed code is an $(n+m,M,\epsilon)$-code that satisfies the following equations (\ref{eq:c1})-(\ref{eq:c3}).

\begin{align}\label{eq:c1}
\log(M+1)\geq n\mathbb{E}\left[\log\left(\frac{p_{Y|X}(Y|X))}{p_{Y}(Y)}\right)\right]-\left({\frac{2+\epsilon}{\epsilon}n\mathtt{{Var}}\left[\log\left(\frac{p_{Y|X}(Y|X))}{p_{Y}(Y)}\right)\right]}\right)^{\frac{1}{2}}
\end{align}

\begin{align}\label{eq:c2}
\log(M+1)&\geq \frac{n}{2}\log(1+\gamma)
-\sqrt{\frac{2+\epsilon}{\epsilon}\frac{\gamma}{1+\gamma}n}-n^{\frac{1}{4}}\\ 
\label{eq:c3}\log(M)&\geq \frac{n}{2}\log(1+\gamma)
-\sqrt{\frac{2+\epsilon}{\epsilon}\frac{\gamma}{1+\gamma}n}-n^{\frac{1}{4}} -1 
\end{align}
Here, (\ref{eq:c2}) is obtained by noting that the mutual information $\mathbb{E}\left[\log\left(\frac{p_{Y|X}(Y|X))}{p_{Y}(Y)}\right)\right]=\frac{1}{2}\log\left(1+\gamma\right)$, while the variance $\mathtt{{Var}}\left[\log\left(\frac{p_{Y|X}(Y|X))}{p_{Y}(Y)}\right)\right]=\frac{\gamma}{1+\gamma}$. The last equation follows by noting that $\log\left(M+1\right)-\log\left(M\right)<1$. Using (\ref{eq:c3}) with the constraints in (\ref{eq:mainc1}) and (\ref{eq:mainc2}) completes the proof.
\subsection*{Proof of Corollary \ref{cor: opt power}}
The proof follows by differentiating (\ref{eq:prop2}) with respect to $P_{\rm{t}}$ and setting $\frac{\partial R^{\infty}_{\rm{EH}}}{\partial P_{\rm{t}}}=0$. This leads to the following equation after simplification. 
\begin{align}\label{eq:proofopt}
\left(P_{\rm{t}}+\sigma^2\right)\log\left(P_{\rm{t}}+\sigma^2\right)&=
\left(1+\log\left(\sigma^2\right)\right)\left(P_{\rm{t}}+\sigma^2\right)+
P_{\rm{E}}\log\left(1+0.5\epsilon\right)-\sigma^2
\end{align}  
With the following change of variables $x=P_{\rm{t}}+\sigma^2$, $c=P_{\rm{E}}\log\left(1+0.5\epsilon\right)-\sigma^2$, and $d=1+\log\left(\sigma^2\right)$, (\ref{eq:proofopt}) can be written as 
$x\log(x)=c+dx$ 
which has the solution
$x=\frac{c}{\text{W}\left[c\exp(-d)\right]}$. Back substituting $x$, $c$, and $d$ in the solution yields (\ref{eq: opt P_t}).

\section*{Appendix B: Multiple Power Beacons}
\subsection*{Energy Supply Probability}
We now derive an exact expression for the energy supply probability in a Poisson network with multiple power beacons. Recall that the harvested energy in a given slot is $Z^{}=P_{\rm{PB}}\mu\sum\limits_{x_k\in\Phi}{\frac{H_k}{ \ell\left(\|x_k\|, \eta\right)}}$. From the definition of the energy supply probability, it follows that
	\begin{align}\label{eq:prop mult1}
P_{\textrm{es}}^{\rm{MP}}\left(m,n,a,\lambda,\eta\right)&=\Pr\left[\sum_{i=1}^{n}X_i^2\leq mZ\right]
	\overset{(a)}{=}\Pr\left[W \leq \frac{mZ}{P_\textrm{t}}\right]\\
	&\overset{(b)}{=}1-\mathbb{E}\left[\sum\limits_{\ell=0}^{\frac{n}{2}-1}\frac{\left(m
	Z\right)^\ell}{(2P_{\rm{t}})^\ell\ell!}e^{-\frac{m}{2P_{\rm{t}}}Z}\right]\\
	&\overset{(c)}{=}1-\sum\limits_{\ell=0}^{\frac{n}{2}-1}(-1)^{\ell}\frac{ {m}^{\ell}}{2^\ell P_t^\ell \ell!}\frac{\text{d}^\ell}{\text{d}s^\ell}\mathcal{L}_Z(s)|_{s=\frac{m}{2P_{\rm{t}}}}
	\end{align}
	where $(a)$ follows by the substitution $W=\sum\limits_{i=1}^{n}\frac{X_i^2}{P_\textrm{t}}$ such that $W$ is a Chi-squared random variable with $n$ degrees of freedom. Equality $(b)$ is obtained by conditioning on the random variable $Z$, and by using the cumulative distribution function of a Gamma random variable (since $W$ can be viewed as a Gamma random variable $\rm{Ga}\left(\frac{n}{2},2\right)$ with shape $\frac{n}{2}$ and scale ${2}$).	Finally, $(c)$ follows from the definition of a Laplace transform of a random variable $X$, namely, $\mathcal{L}_X(s)=\mathbb{E}[e^{-sX}]$, and by invoking the property $\mathbb{E}[X^\ell e^{-sX}]=(-1)^\ell\frac{\text{d}^\ell}{\text{d}s^\ell}\mathcal{L}_X(s)$. 

\subsection*{Proof of Lemma \ref{lemma:laplace}}
We now derive the Laplace transform $\mathcal{L}_Z(s)$ for the path loss model $\ell\left(r,\eta\right)\triangleq\max(1,r^\eta)$ considered in Lemma \ref{lemma:laplace}. 
\allowdisplaybreaks{
\begin{align}
\mathbb{E}\left[e^{-sZ}\right]&=\mathbb{E}\left[e^{-sP_{\rm{PB}}\mu\sum\limits_{x_k\in\Phi}{\frac{H_k}{ \ell\left(\|x_k\|,\eta\right)}}}\right]=
\mathbb{E}\left[\prod\limits_{x_k\in\Phi}e^{-sP_{\rm{PB}}\mu{\frac{H_k}{ \ell\left(\|x_k\|,\eta\right)}}}\right]
\nonumber\\
&=\mathbb{E}_\Phi\left[\prod\limits_{x_k\in\Phi}\mathbb{E}_{H_k}\left[e^{-sP_{\rm{PB}}\mu{\frac{H_k}{ \ell\left(\|x_k\|,\eta\right)}}}\right]\right]\overset{(a)}{=}
\mathbb{E}_\Phi\left[\prod\limits_{x_k\in\Phi}\frac{1}{1+sP_{\rm{PB}}\mu\ell\left(\|x_k\|,\eta\right)^{-1}}\right]
\nonumber\\
&\overset{(b)}{=}\exp\left(-2\pi\lambda\int\limits_{0}^{1}\left[1-\frac{1}{1+sP_{\rm{PB}}\mu}r\text{d}r\right]-
2\pi\lambda\int\limits_{1}^{\infty}\left[1-\frac{1}{1+sP_{\rm{PB}}\mu r^{-\eta}}r\text{d}r\right]\right)
\nonumber\\
&\overset{(c)}{=}\exp\left(-\pi\lambda\frac{{s}P_{\rm{PB}}\mu}{1+{s}P_{\rm{PB}}\mu}\right)
\exp\left(-\pi\lambda~\mathcal{F}\left({s}P_{\rm{PB}}\mu,\eta\right)\right)
\end{align}}where $(a)$ follows from the independence of small-scale fading across the PB-EH links, and by further conditioning on the locations of the PB nodes, and $(b)$ is obtained by invoking the probability generating functional (PGFL) of the PPP $\Phi$ \cite{haenggi2012stochastic}. Finally, $(c)$ results by expressing the integrals in terms of the hypergeometric function as defined in (\ref{eq:functionF}).

\subsection*{Proof of Lemma 3}
The proof follows by noting that $\mathbb{E}[H_k]=1$, and by applying Campbell's theorem\cite{haenggi2012stochastic} to obtain 
$\mathbb{E}\left[Z\right]=P_{\rm{PB}}\mu 2\pi\lambda\left(\int_{0}^{1}r\text{d}r + 
\int_{1}^{\infty}r^{1-\eta}\text{d}r\right)=
P_{\rm{PB}}\mu\lambda\pi\frac{\eta}{\eta-2}$.
\subsection*{Achievable Rate}
The achievable rate for the case of multiple power beacons can be derived following the procedure in Appendix A. Similar to (\ref{eq:eneergyBound1}), we first bound the energy outage probability as
\begin{align}\label{eq:eneergyBound2}
\Pr\left[{\bigcup_{k=1}^{n}}\left\{\sum_{\ell=1}^{k}X_\ell^2\geq\sum_{i=1}^{m}Z_i\right\} \right]=
\sum\limits_{i=0}^{\frac{n}{2}-1}(-1)^{i}\frac{{m}^{i}}{(2a)^ii!}\frac{\text{d}^i}{\text{d}s^i}\mathcal{L}_Z(s)|_{s=\frac{m}{2a}}
&\leq \frac{\epsilon}{2+\epsilon}
\end{align}
where we have used the expression (and the notation) from Proposition \ref{lemma: multi_pec}. Following steps similar to (\ref{eq:rule})-(\ref{eq:c2}), we recover the result presented in Theorem \ref{th:MP}.

\bibliographystyle{ieeetr}
\end{document}